\documentclass[10pt,journal,compsoc]{IEEEtran}

\usepackage{amssymb,amsfonts}
\usepackage{amsmath}
\usepackage{graphicx}
\usepackage[ruled,vlined,linesnumbered]{algorithm2e}
\usepackage{caption}
\usepackage{mathrsfs}
\usepackage{multirow}
\usepackage{graphicx}
\usepackage{pdfpages}
\usepackage{pgfplots}
\pgfplotsset{compat=1.3}
\usepackage{flowchart}
\usetikzlibrary{arrows}
\usepackage{subfigure}
\usepackage[outdir=./]{epstopdf}
\usepackage{url}
\usepackage{amsfonts}
\usepackage{times}
\usepackage{pgf}
\usepackage{tikz}
\usetikzlibrary{arrows,automata,positioning}
\usepackage[latin1]{inputenc}
\usepackage{verbatim}
\usepackage{makecell}
\usepackage{booktabs}
\usepackage{pifont}
\usepackage{hyperref}
\usepackage{adjustbox}
\usepackage{framed}

\usepackage{xcolor}

\newcommand{\db}[1]{[\kern-0.15em[ #1
 ]\kern-0.15em]}
\newcommand{\emp}{\langle\rangle}
\makeatletter
\newcommand*{\rom}[1]{\expandafter\@slowromancap\romannumeral #1@}
\makeatother
\newcommand{\specialcell}[2][c]{%
  \begin{tabular}[#1]{@{}c@{}}#2\end{tabular}}
\newcommand{\xmark}{\text{\ding{55}}}
\newcommand{\mypara}[1]{\vspace{1.5ex}\noindent\emph{#1}\hspace{1.5ex}}
\newtheorem{remark}{Remark}

\ifCLASSOPTIONcompsoc
  \usepackage[nocompress]{cite}
\else
  \usepackage{cite}
\fi

\ifCLASSINFOpdf
\else
\fi

\begin{document}
\title{
Automatically `Verifying' Discrete-Time Complex Systems through Learning, Abstraction and Refinement
}

\author{Jingyi~Wang,
        Jun~Sun,~Shengchao~Qin,
        and~Cyrille~Jegourel
\IEEEcompsocitemizethanks{\IEEEcompsocthanksitem Jingyi Wang is with College of Computer Science and Software Engineering, Shenzhen University, China. He is also with Singapore University of Technology and Design, Singapore. E-mail: jingyi\_wang@sutd.edu.sg \protect
\IEEEcompsocthanksitem Jun Sun, the corresponding author, is with Singapore University of Technology and Design, Singapore. E-mail: sunjun@sutd.edu.sg.\protect
\IEEEcompsocthanksitem Shengchao Qin is with School of Computing, Media and the Arts, Teesside University, UK. He is also with College of Computer Science and Software Engineering, Shenzhen University, China. E-mail: S.Qin@tees.ac.uk
\IEEEcompsocthanksitem Cyrille Jegourel is with Singapore University of Technology and Design, Singapore. E-mail: cyrille\_jegourel@sutd.edu.sg
\IEEEcompsocthanksitem This work was supported in part by National Natural Science Foundation of China under Grant No. 61772347 and Science and Technology Foundation of Shenzhen City under Grant No. JCYJ20170302153712968.

}
}


\IEEEtitleabstractindextext{%
\begin{abstract}
Precisely modeling complex systems like cyber-physical systems is challenging, which often render model-based system verification techniques like model checking infeasible. To overcome this challenge, we propose a method called LAR to automatically `verify' such complex systems through a combination of learning, abstraction and refinement from a set of system log traces. We assume that log traces and sampling frequency are adequate to capture `enough' behaviour of the system. Given a safety property and the concrete system log traces as input, LAR automatically learns and refines system models, and produces two kinds of outputs. One is a counterexample with a bounded probability of being spurious. The other is a probabilistic model based on which the given property is `verified'. The model can be viewed as a proof obligation, i.e., the property is verified if the model is correct. It can also be used for subsequent system analysis activities like runtime monitoring or model-based testing. Our method has been implemented as a self-contained software toolkit. The evaluation on multiple benchmark systems as well as a real-world water treatment system shows promising results.

%
%
%
\end{abstract}

}

\maketitle

\IEEEdisplaynontitleabstractindextext
\IEEEpeerreviewmaketitle

\section{Introduction}\label{intro}
\IEEEPARstart{C}{yber-physical} systems (CPS) integrate physical and engineered systems and have the potential to transform the way people interact with engineered systems. They are often used to control public infrastructures like water purification/distribution systems or smart grid systems. When CPS are employed in such safety-critical scenarios, it is desirable to show that they can operate dependably and safely. Analyzing CPS, however, is challenging. Existing system analysis methods, like model-based testing, model checking and theorem proving, require the availability of a system model. 
Because CPS closely interact with the physical environment, the model must not only capture the system behavior but also the environment's. Modeling the environment is often hard, due to complicated continuous dynamics in the physical environment. 

CPS are merely an example of those complex systems for which manual modeling is challenging. To tackle the challenge, multiple approaches that do not rely on manual modeling have been explored. One example is statistical model checking (SMC)~\cite{YounesThesis}. The idea is to provide a statistical measure on the likelihood of satisfying a given property, by observing sample system traces and applying techniques like hypothesis testing~\cite{havelund2002synthesizing,YounesThesis}. 
However, SMC has its limitations. For instance, since SMC relies on sampling \emph{finite} system traces, it is challenging to verify un-bounded properties~\cite{younes2011statistical,rohr2013simulative}, like eventually always something good happens. \emph{Furthermore, when SMC claims a property is verified, it provides no insight on why the property holds.} If a new property is given, SMC must be applied from scratch.
Another approach for avoiding manual modeling is to automatically learn models from system logs (i.e., sample system traces). Multiple learning algorithms have been proposed to learn a variety of models, e.g.,~\cite{sen2004learning,ron1996power,carrasco1994learning,de2010grammatical}. It has been shown that such learned models can be useful for system analysis in certain scenarios. Recently, the idea has been extended to learn models for system verification through model checking. In~\cite{AA,AAJ,chen2012learning,mao2012learning}, the authors proposed to learn probabilistic models and then apply techniques like probabilistic model checking (PMC) to calculate the probability of satisfying a property based on the learned model. Compared to SMC, learning could be beneficial as it overcomes several known limitations with SMC. For instance, we can verify unbounded properties based on the learned model. Furthermore, the learned model could be useful for a range of system analysis or control objectives, e.g., for model-based testing and for implementing runtime monitors.

Existing learning methods~\cite{sen2004learning,ron1996power,carrasco1994learning,de2010grammatical,AA,chen2012learning,mao2012learning} however have multiple issues. Firstly, they are designed with a fixed level of abstraction, which limits their applicability to real-world systems. For instance, the traces we obtain from a real-world water treatment system~\cite{swat} capture the reading of 25 sensors plus 26 variables used in the control software. Furthermore, these variables are mostly float or double. Without abstraction, it is hard to learn a precise and reasonably small model of the system. Determining the right level of abstraction is however highly non-trivial. As far as we know, it has not been investigated on how to learn probabilistic models at the right level of abstraction. Secondly, existing learning methods do not take into account the property to be verified. In a recent empirical study~\cite{wang2017should}, it is observed that the verification results based on the learned models could deviate significantly from the actual results.


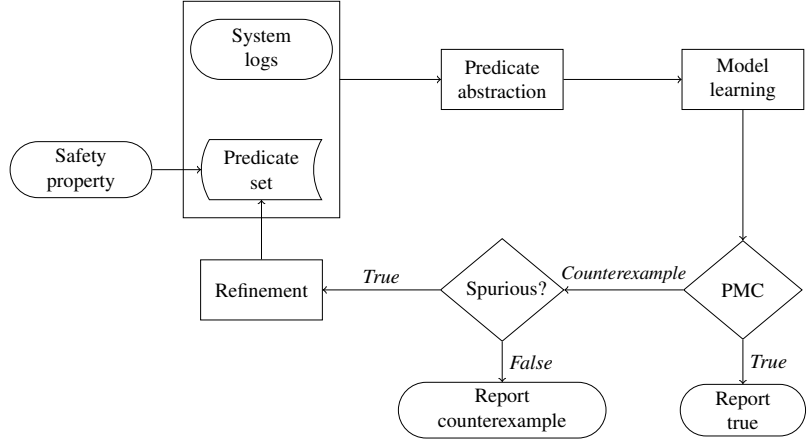
\begin{figure*}[]
	\centering
	\scalebox{.8}{
\begin{tikzpicture}
\node (sp) at (1,0) [draw, terminal, minimum width=0.5cm, text centered, text width=1.2cm, minimum height=0.5cm] {Safety property};

\node (ps) at (4,0) [draw, storage, minimum width=2cm, text centered, text width=1.5cm, minimum height=1cm] {Predicate set};

\node (sl) at (4,2) [draw, terminal, minimum width=2cm, text centered, text width=1.2cm, minimum height=1cm] {System logs};

\node (pa) at (8,1.5) [draw, process, minimum width=2cm, text centered, text width=1.8cm, minimum height=1cm] {Predicate abstraction};

\node (ml) at (12,1.5) [draw, process, minimum width=2cm, text centered, text width=1.8cm, minimum height=1cm] {Model learning};

\node (mc) at (12, -2) [draw, diamond, aspect=1.2, minimum width=0.5cm, text centered, text width=1.2cm, minimum height = 0.5cm] {PMC};

\node (ht) at (8, -2) [draw, diamond, aspect=1.2, minimum width=0.5cm, text centered, text width=1.2cm, minimum height = 0.5cm] {Spurious?};

\node (rf) at (4,-2) [draw, process, minimum width=2cm, text centered, text width=1.8cm, minimum height=1cm] {Refinement};

\node (rc) at (8,-4) [draw, terminal, minimum width=0.5cm, text centered, text width=2.3cm, minimum height=0.5cm] {Report counterexample};

\node (rt) at (12,-4) [draw, terminal, minimum width=0.4cm, text centered, text width=1cm, minimum height=0.5cm] {Report true};

\draw (2.7,-0.8) rectangle (5.3,2.8);

\draw[->] (5.3,1.5) -- (pa);
\draw[->] (sp) -- (ps);
\draw[->] (pa) -- (ml);
\draw[->] (ml) -- (mc);
\draw[->] (mc) -- node [above] {\small{\textit{Counterexample}}} (ht);
\draw[->] (mc) -- node [right] {\textit{True}}  (rt);
\draw[->] (ht) -- node [right] {\textit{False}} (rc);
\draw[->] (ht) -- node [above] {\textit{True}} (rf);
\draw[->] (rf) -- (ps);
 \end{tikzpicture}}
 \caption{An overview of our framework}
 \label{fig:workflow}
 \end{figure*}

In this work, we aim to develop a method to `verify' CPS by learning models at a level of abstraction which is ideal for verifying or falsifying the given property, and by generating verification results which are validated against the actual system. The method we propose is called LAR, which is a novel  combination of probabilistic model \emph{l}earning and counterexample guided \emph{a}bstraction \emph{r}efinement (CEGAR)~\cite{cegar,clarke2003counterexample,probcegar}. Combing learning and CEGAR is far from straightforward, due to the lack of a ``ground-truth'' model.

The overall workflow is shown in Fig.~\ref{fig:workflow}. The input of LAR includes a probabilistic property and a set of system traces, which could be obtained from a logger in the system. We assume that the logging mechanism (i.e., the logged variables and the logging frequency) are adequate to capture `enough' behaviour of the system regarding the property. We first construct a set of abstract system traces through predicate abstraction. 
Next, we apply existing automatic model learning techniques to construct a probabilistic model of the system, in the form of a discrete time Markov chain (DTMC). 
Afterwards, we apply PMC to verify the model against the property. If a counterexample is identified, we check whether the counterexample, in the form of a set of paths of the learned model, is spurious. Notice that because we do not have a model of the system, we cannot completely verify whether the counterexample is spurious or not. We solve the problem by applying hypothesis testing to bound the probability of the counterexample being spurious. If the counterexample is not spurious, we report that the system fails the property. Otherwise, we analyze the counterexample to generate a new predicate which would rule out the counterexample. 

Due to the lack of a system model, standard methods for generating new predicates (e.g., weakest pre-condition calculation~\cite{clarke2003counterexample} or interpolation~\cite{henzinger2004abstractions}) are infeasible. We solve the problem by adopting classification techniques from the machine learning community (refer to Section~\ref{refinement} for details).  We then repeat the process until either we have identified a counterexample or have constructed a probabilistic model of the system based on which the property is verified. 
In the latter case, the model generated by our method could be viewed as a proof obligation (i.e., the property is verified if the model is correct), which could be further evaluated by experts or through other means (like statistical validation of the model). Furthermore, our learned models can be potentially used for many subsequent system analysis activities. For instance, we could use the models for model checking~\cite{kwiatkowska2011prism} and implementing runtime monitors~\cite{sistla2011runtime}. The models can also help design more robust systems~\cite{karsai2008model}. We could also use them for model-based testing to generate more useful test cases~\cite{pretschner2005model}. Lastly, our method has the potential to help people understand how a complex system works by automatically extracting relevant predicates. We implemented our method as a software toolkit and applied it to benchmark systems and a real-world water treatment system. The evaluation validated our approach well. 

The remainder of the paper is organized as follows. Section~\ref{back} reviews necessary background and defines our problem. Section~\ref{approach} presents the details of our approach. Section~\ref{evaluation} presents our implementation and evaluates our method. Section~\ref{related} reviews related work and Section~\ref{con} concludes.

\section{Background}\label{back}
\mypara{System Assumptions.} We assume that the system under analysis, denoted as $M$, has $n$ observable variables, i.e., $V_M=\{V_1,V_2,\cdots,V_n\}$. Each variable $V_i$ is associated with a domain $D_i$ which may have infinite values. For instance, the variables in a water treatment system include those variables in the control software (with finite domains) as well as those representing the environment (with infinite domains).
We assume that $M$ is a complicated system such that we do not know exactly how the variables $V_M$ change. However, we assume that $M$ is deterministic if $V_M$ is fully observed, i.e., the same sequence valuations of $V_M$ lead to the same system behaviors. The state of $M$ can be observed through implementing a logger in the system which (e.g., periodically) outputs the valuation of $V_M$. Thus, we can obtain a set of finite traces of the system. In the following, we assume that the logged variables and the sampling frequency are adequate to capture the relevant behavior of the system.

We write $\Sigma(V_M)$ (hereafter $\Sigma$ for short) to denote $D_1 \times \cdots \times D_n$, which is the set of all observable states of $M$. Note that $\Sigma$ may be infinite if there are float variables with infinite domains. A finite system trace is a sequence $\pi=\langle s_1,s_2,\cdots,s_n \rangle$ where $s_i\in\Sigma$ for all $i$. In other words, a finite system trace can be seen as a finite string over $\Sigma$. We denote the set of all finite strings over $\Sigma$, including the empty string $\emp$, as $\Sigma^{\star}$. Given a string $t$, we say that a string $t'$ is a prefix of $t$ if and only if there is a string $t'' \in \Sigma^{\star}$ such that $t = t' \cdot t''$ where $\cdot$ denotes string concatenation. We write $\mathit{prefix}(t)$ to denote the set of all prefixes of $t$. We denote a system trace with length $l$ by $\pi^l$ and its $k$-th letter by $\pi[k]$. We also use $\pi[i\cdots j]$ to denote the trace from $\pi[i]$ to $\pi[j]$.

\mypara{System model.} Let $\Sigma_0$ be the set of initial states of system $M$. If we impose a prior probability distribution $\mu_0$ on the states in $\Sigma_0$ (e.g., a uniform distribution over all states in $\Sigma_0$), $M$ can be effectively viewed as a discrete time Markov Chain (DTMC)~\cite{pmc}. Formally, 
\\
\newtheorem{definition}{Definition}
\begin{definition}
A DTMC $\mathcal{D}$ is a tuple $(S,\imath_{init},Pr)$, where $S$ is a countable, nonempty set of states
; $\imath_{init}: S \to [0,1]$ is the initial distribution s.t.~$\sum_{\mathrm{s}\in S}\imath_{init}(\mathrm{s})=1$; and $Pr: S \times S \to [0,1]$ is the transition probability assigned to every pair of states which satisfies the following condition: $\sum_{\mathrm{s}'\in S}Pr(\mathrm{s},\mathrm{s}')=1$. 
\end{definition}~\\
\noindent An example DTMC is shown in Fig.~\ref{fig:dtmc}, where states are labeled for readability. A DTMC induces an underlying digraph where states act as vertices and there is an edge from $\mathrm{s}$ to $\mathrm{s}'$ if and only if $Pr(\mathrm{s},\mathrm{s}')>0$. Given a path $\pi = \langle \mathrm{s}_1, \mathrm{s}_2, \cdots, \mathrm{s}_n \rangle$ in $\mathcal{D}$, we write $\mathcal{P}(\pi, \mathcal{D}) = Pr(\mathrm{s}_1, \mathrm{s}_2) \times Pr(\mathrm{s}_2, \mathrm{s}_3) \times \cdots \times Pr(\mathrm{s}_{n-1}, \mathrm{s}_n)$ to denote the probability of exhibiting $\pi$ in $\mathcal{D}$. Furthermore, we write $\mathit{Path}_{fin}(\mathcal{D})$ to denote the set of finite paths of $\mathcal{D}$ starting with an initial state, i.e., a state $\mathrm{s}$ such that $\imath_{init}(\mathrm{s}) > 0$.

$M$ can be viewed as a DTMC $\mathcal{D}_M = (S,\imath_{init},Pr)$ where $S = \Sigma$ is the states of $M$; $\imath_{init} = \mu_0$ is the imposed prior probability distribution of the initial states.

 \begin{figure}[t]
	 \centering
 \begin{tikzpicture}[->,>=stealth',shorten >=1pt,auto]

\node[initial,state] (s1) {$a$};
\node[state] (s3) [right = 1.5cm of s1] {$c$};
\node[state] (s2) [below = 1.2cm of s1] {$b$};
\node[state] (s4) [right = 1.5cm of s2] {$d$};

\path (s1) edge [loop above] node {0.98} (s1)
edge [pos=0.5] node {0.01} (s2)
edge node {0.01} (s3)
(s2) edge [loop left] node {0.5} (s2)
edge [] node {0.5} (s4)
(s3) edge [loop above] node {0.9} (s3)
edge node {0.1} (s4)
(s4) edge [loop right] node {1} (s4)
;
\end{tikzpicture}
 \caption{An example DTMC}
  \label{fig:dtmc}
 \end{figure}
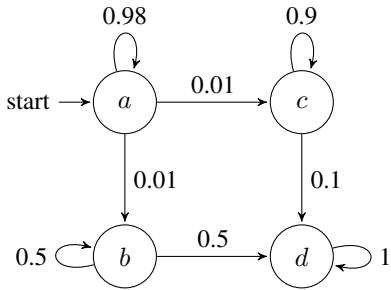

\mypara{Model learning.} Once we obtain a set of system traces from the system logger, we can then apply a probabilistic model learning algorithm to `learn' a DTMC. Most existing model learning algorithms are inspired by ALERGIA~\cite{carrasco1994learning}, which first transforms the traces into a prefix tree and then iteratively merges the tree nodes based on whether the two nodes are compatible. Based on the same idea, a genetic-algorithm based learning algorithm is proposed in~\cite{wang2017should} to merge the tree nodes through evolution.
We present one of the latest ALERGIA-style learning algorithms as a representative~\cite{AA,AAJ} in Section~\ref{learning-algo}. With the learned DTMC, we can verify safety properties by probabilistic model checking.

\mypara{Probabilistic model checking.} Assume that we are given a property in the form of a probabilistic linear temporal logic formula (PLTL)~\cite{hansson1994logic,baier2003comparative}. The syntax of the PLTL is defined as follows.
\begin{displaymath}
	\varphi::=p|\neg \varphi | \varphi \land \varphi | \mathcal{P}_{\bowtie r}(\varphi \textbf{U} \varphi)
\end{displaymath}
where $p$ is an atomic proposition constituted by $V$, $\ r\in[0,1]\subseteq \mathcal{R}$, $\bowtie\in\{<,\leq,>,\ge\}$ and \textbf{U} is a modal operator which reads `until'. The properties of interest in this work is the safety fragment of PLTL properties. In the following, we focus on the property of the form\footnote{Recursive PLTL properties are not handled due to the limit of model learning~\cite{AAJ}.}: $\mathcal{P}_{\leq r}(\textbf{F}\varphi)$ where $\textbf{F} \varphi$ is a short form for $true \textbf{U} \varphi$.
Intuitively, this property states that the probability of eventually reaching a bad state does not exceed a certain threshold.  For instance, property $\mathcal{P}_{\leq 0.2}(\textbf{F}\ observe0 > 1)$ states that the property of reaching a state with $observe0 > 1$ is less than 0.2. A DTMC $\mathcal{D}$ satisfies $\mathcal{P}_{\leq r}(\textbf{F}\varphi)$ if and only if the accumulated probability of all paths in $\mathit{Path}_{fin}(\mathcal{D})$ which satisfy $\textbf{F}\varphi$ is less than or equal to $r$. We write $\mathcal{D} \vDash \mathcal{P}_{\leq r}(\textbf{F}\varphi)$ to denote that DTMC $\mathcal{D}$ satisfies $\mathcal{P}_{\leq r}(\textbf{F}\varphi)$.

The reason that we focus on probabilistic properties is that safety-critical CPS are often designed with built-in mechanisms for handling safety violation (e.g., a shutdown sequence is triggered once a violation is detected). The goal of analysis is thus to show such safety violation is rare (i.e., with low probability). For instance, in our case study of the water treatment system, a safety property is that the water level in the backwash tank must be within certain range. Otherwise, a system shutdown is triggered. Our task is then to show the probability of the water level being out of the range is low enough such that the probability of triggering system shutdown is low. \\

\noindent \textbf{Problem Definition.} Our problem is then defined as follows. Given a system $M$ and a property $\mathcal{P}_{\leq r}(\textbf{F}\varphi)$, how do we check whether the property is satisfied by $M$ without having its model $\mathcal{D}$? In case that the property is not satisfied, can we present some evidence , i.e., in the form of counterexamples? In case that the property is satisfied, can we present some evidence as well?

\section{Our Approach} \label{approach}
Since the system is complicated and a precise system model is unavailable, we cannot apply existing techniques like PMC to check whether the property is satisfied or not.
One way to solve the problem is to construct a DTMC model $\mathcal{D}$ approximating $M$ and then verify the given property based on $\mathcal{D}$. There are however a number of questions that must be answered in order to make this approach work. First, how do we construct $\mathcal{D}$ systematically? In particular, what are the states in $\mathcal{D}$ and what are the transition probabilities? If we consider every different valuation of variables $V_M$ to be a different state, $\mathcal{D}$ is likely to contain many (and often infinitely many) states. Such $\mathcal{D}$ is difficult to learn or verify. We thus would like to construct a small $\mathcal{D}$ (with states more abstract than a valuation of $V_M$) which would be used to verify the property. Secondly, after constructing $\mathcal{D}$ and verifying the property based on $\mathcal{D}$, how do we quantify the confidence we have on the verification result, knowing that $\mathcal{D}$ may not be precise? In the following, we present answers of these questions. For simplicity, we illustrate how our approach works using the following running example (which is motivated by the \emph{crowds} protocol~\cite{RR98}). \\

\noindent \textbf{Example} Assume a system where multiple users are browsing the web (i.e., sending/receiving messages to/from Web servers). Further assume there are eavesdroppers on the network who can observe the \emph{direct} source of a message. In order to provide anonymity for the users, a message from a user is not directly sent to the destination, but rather routed among the users so that the eavesdroppers cannot identify the actual source of the message. Each user in the system uses a complicated algorithm to decide on whether to forward a received message to its destination or to some other user on the network. Assume there are a total of $u$ users. The property to be verified is that the probability of observing a message sending to its destination by its actual source more than once in $r$ runs should not exceed a threshold, say $0.2$. Formally, it is specified as: ${\cal P}_{\leq\ 0.2}(\textbf{F}\ observe0>1)$ where $observe0$ is a variable in the system which captures the number of message sending to its destination by its actual source. Without knowing how each user decides on forwarding the messages, we cannot develop a precise model.

\begin{table}[]
\centering
\caption{Details of the variables in the $crowds$ protocols.}
\label{crowds-vars}
\begin{tabular}{@{}lll@{}}
\toprule
Name                  & Type & Meaning                              \\ \midrule
launch                & bool & Start modeling?                      \\
new                   & bool & Initialize a new protocol instance?  \\
runCount              & int  & Counts protocol instances            \\
start                 & bool & Start the protocol?                  \\
run                   & bool & Run the protocol?                    \\
lastSeen              & int  & Last crowd member to touch msg       \\
good                  & bool & Crowd member is good?                \\
bad                   & bool & Crowd member is bad?                 \\
recordLast            & bool & Record last seen crowd member?       \\
badObserve            & bool & Bad members observes who sent msg?   \\
deliver               & bool & Deliver message to destination?      \\
done                  & bool & Protocol instance finished?          \\
observe0 to observe19 & int  & Counters for attackers' observations \\ \bottomrule
\end{tabular}
\end{table}

\subsection{Predicate Abstraction}
In our approach, we start with collecting a set of system traces $\Pi$, by introducing a logger in the system and executing the system multiple times. 
For instance, in the running example, we log the valuation of all 32 variables, including where the message is originated, where it is forwarded to, etc. The details of all the logged variables are shown in Table~\ref{crowds-vars}. The logged traces contain many details, most of which may not be relevant to verifying the property. Thus, we start with abstracting the traces, through predicate abstraction~\cite{wachter2007probabilistic}. Recall that a concrete trace is a sequence $\langle s_0, s_1, \cdots, s_n \rangle$ where $s_i \in \Sigma$ for all $i$ is a valuation of variables $V_M$. Predicate abstraction is to construct an abstract trace $\langle a_0, a_1, \cdots, a_n \rangle$ where, for all $i$, $a_i$ is the valuation of a set of propositions $AP$ given $s_i$.


Let $\mathit{BExpr}_V$ denote the Boolean expressions over $V$. A proposition $\psi$ is a Boolean expression over a set of variables. For an expression $e \in \mathit{Expr}_{V}$, we denote its valuation in state $s \in \Sigma$ by $\db{e}_s$. For a Boolean expression $e$, $\db{e}_s\in\{0,1\}$ where 0 stands for \emph{false} and 1 for \emph{true}. We write $s\models\psi$ iff $\db{\psi}_s=1$. We denote the set of states satisfying a predicate $\psi$ by $\db{\psi}=\{s|s\in\Sigma \land s\models\psi\}$.

Let $P =\{p_1,\cdots,p_k\} \subseteq \mathit{BExpr}_V$ be a set of propositions over $V$. Given a state $s \in \Sigma$, we define an abstraction function as: $\alpha_P(s)=(\db{p_1}_s,\cdots,\db{p_k}_s)$, which maps the state $s$ to an abstract state, i.e., a bit vector with length $k$ where each bit represents the truth value of a proposition in $P$.
We write $\Sigma_P$ to denote the set of abstract states with respects to $P$. 
Given a logged trace $\pi = \langle s_1, s_2, \cdots, s_n \rangle$, we can construct an abstract trace with respect to predicates $P$ as: $\pi_P = \langle \alpha_P(s_1), \alpha_P(s_2), \cdots, \alpha_P(s_n) \rangle$.
Let $\Pi$ be the set of sample traces such that each $\pi \in \Pi$ is a string in $\Sigma^*$. We can construct a set of abstract traces $\Pi_P$ from $\Pi$ by abstracting each trace in $\Pi$ one by one.

In our running example, we assume that the set $\Pi$ contains two logged traces: $\langle 0,0,1,1 \rangle$ and $\langle 0,1,2,2 \rangle$ where each number denotes a value of variable $observe0$. Note that we have removed the values of all other 32 variables for simplicity. 
Based on the above-mentioned property, we set the initial set of predicates $P$ to be $\{observe0>1\}$. After predicate abstraction, we obtain two corresponding abstract traces: $\langle 0,0,0,0 \rangle$ and $\langle 0,0,1,1 \rangle$ where a number $0$ means that $observe0\leq1$ and 1 means $observe0>1$.


\subsection{Model Learning}\label{learning-algo}
With the set of abstract traces $\Pi_P$, we then apply existing model learning techniques~\cite{AA,AAJ,chen2012learning,wang2017should} to construct a DTMC model $\mathcal{D}_P$. 
The essential idea of model learning is to construct a DTMC $\mathcal{D}_P$ such that it maximizes the probability of observing the traces $\Pi_P$ as well as contains a relatively small number of states. Several learning algorithms have been proposed~\cite{sen2004learning,ron1996power,carrasco1994learning,de2010grammatical,wang2017should}. In the following, we present one of the learning algorithms called AALERGIA~\cite{AA,AAJ} as a representative and remark that our approach can be configured to work with different learning algorithms, as we show in Section~\ref{evaluation}.

\mypara{AALERGIA learning algorithm} The AALERGIA algorithm makes the following assumptions. First, the sample traces are generated by a system which can be modeled as a DTMC. 
Secondly, the sample traces are mutually independent. Thirdly, the length of each trace is independent from the sequence of states observed. The AALERGIA algorithm has been proved to converge to the actual model if a sufficiently large number of sample traces are provided~\cite{AA,AAJ}. 


Let $\mathit{prefix}(\Pi_P) = \{\mathit{prefix}(t) | t \in \Pi_P\}$ be the set of all prefixes of any trace in $\Pi_P$. The set of traces $\Pi_P$ can be naturally organized into a tree $\mathit{tree}(\Pi_P) = (N, root, E)$ where each node in $N$ is a member of $\mathit{prefix}(\Pi_P)$; the root is the empty string $\langle \rangle$; and $E \subseteq N \times N$ is a set of edges such that $(\pi, \pi')$ is in $E$ if and only if there exists $e \in \Sigma$ such that $\pi \cdot \langle e \rangle = \pi'$. For instance, assume that  $\Pi_P$ in our running example contains 100 abstract traces: 88 of them are $\langle 0,0,0,0 \rangle$, 2 of them are $\langle 0,0,0,1 \rangle$, 2 of them are $\langle 0,0,1 \rangle$, and 8 of them are $\langle 0,0,1,1 \rangle$. Fig.~\ref{fig:tree0} shows the tree representing this set of abstract traces, where each node represents a prefix of some abstract trace. The state labels are to be explained later.


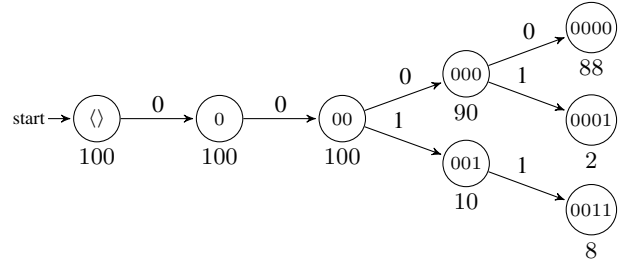
\begin{figure}[t]
\centering
{\scriptsize\begin{tikzpicture}[>=stealth',shorten >=1pt,auto,every node/.style={minimum size=1.5em,inner sep=1}]
	
  \node[initial,state,label=below:{\small $100$}] (A0)      {$\langle\rangle$};
  \node[state,label=below:{\small $100$}]         (A1) [right = 1cm of A0]  {$0$};
  \node[state,label=below:{\small $100$}] 		   (A2) [right = 1cm of A1]  {$00$};
  \node[state,label=below:{\small $90$}]         (B0) [above right = 0.15cm and 1.2cm of A2] {$000$};
  \node[state,label=below:{\small $88$}]         (B1) [above right = 0.15cm and 1.2cm of B0] {$0000$};
  \node[state,label=below:{\small $2$}]         (B2) [below right = 0.15cm and 1.2cm of B0] {$0001$};

  \node[state,label=below:{\small $10$}]         (C0) [below right = 0.15cm and 1.2cm of A2] {$001$};
  \node[state,label=below:{\small $8$}]         (C1) [below right = 0.15cm and 1.2cm of C0] {$0011$};
  
  \path[->] (A0)  
             edge              node {\small 0} (A1)
			 
			(A1)  edge node {\small 0} (A2)
			
			(A2)  edge [pos=0.7] node {\small 0} (B0)
				  edge [pos=0.25] node {\small 1} (C0)

        	(B0) edge [pos=0.7] node {\small 0} (B1)
				 edge [pos=0.25] node {\small 1} (B2)

			(C0) edge [pos=0.25] node {\small 1} (C1);
			           
\end{tikzpicture}}
\caption{Example tree representation of samples}
\label{fig:tree0}
\end{figure}

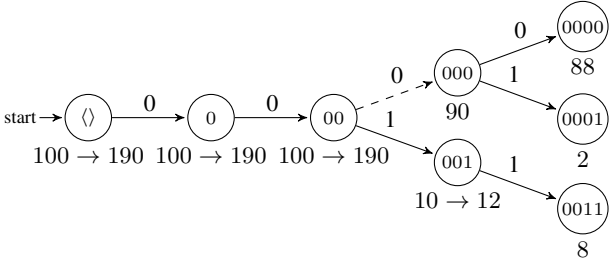
\begin{figure}[t]
\centering
{\scriptsize\begin{tikzpicture}[>=stealth',shorten >=1pt,auto,every node/.style={minimum size=1.5em,inner sep=1}]
	
  \node[initial,state,label=below:{\small $100\to 190$}] (A0)      {$\langle\rangle$};
  \node[state,label=below:{\small $100\to 190$}]         (A1) [right = 1cm of A0]  {$0$};
  \node[state,label=below:{\small $100\to 190$}] 		   (A2) [right = 1cm of A1]  {$00$};
  \node[state,label=below:{\small $90$}]         (B0) [above right = 0.15cm and 1.2cm of A2] {$000$};
  \node[state,label=below:{\small $88$}]         (B1) [above right = 0.15cm and 1.2cm of B0] {$0000$};
  \node[state,label=below:{\small $2$}]         (B2) [below right = 0.15cm and 1.2cm of B0] {$0001$};

  \node[state,label=below:{\small $10\to 12$}]         (C0) [below right = 0.15cm and 1.2cm of A2] {$001$};
  \node[state,label=below:{\small $8$}]         (C1) [below right = 0.15cm and 1.2cm of C0] {$0011$};
  
  \path[->] (A0)  
             edge              node {\small 0} (A1)
			 
			(A1)  edge node {\small 0} (A2)
			
			(A2)  edge [dashed,pos=0.7] node {\small 0} (B0)
				  edge [pos=0.25] node {\small 1} (C0)

        	(B0) edge [pos=0.7] node {\small 0} (B1)
				 edge [pos=0.25] node {\small 1} (B2)

			(C0) edge [pos=0.25] node {\small 1} (C1);
			           
\end{tikzpicture}}
\caption{Merge node $00$ and node $000$.}
\label{fig:tree}
\end{figure}

The AALERGIA algorithm is inspired by stochastic regular grammatical inference, which aims to learn the structure of a stochastic finite state automaton and estimate its transition probabilities~\cite{kermorvant2002stochastic}. The idea is to generalize $\mathit{tree}(\Pi_P)$ by merging the tree nodes according to certain criteria in certain fixed order. Intuitively, two nodes should be merged if they are likely to represent the same state in the underlying DTMC. Since by assumption we do not know the underlying DTMC, whether two nodes should be merged is heuristically decided through a \emph{compatibility test}, which intuitively measures how similar two nodes are.  We remark that the compatibility test effectively controls the degree of generalization. Different types of compatibility test have been studied~\cite{carrasco1994learning,ron1995learnability,kermorvant2002stochastic}. In~\cite{carrasco1994learning}, the compatibility test is based on the Hoeffding bounds, whereas the DSAI algorithm proposed in~\cite{de2010grammatical} uses a condition based on calculating certain distance between the two nodes. In the following, we present the compatibility test adopted in AALERGIA. 

First, each node $\pi$ in $\mathit{tree}(\Pi_P)$ is labeled with the number of traces $\pi'$ in $\Pi_P$ such that $\pi$ is a prefix of $\pi'$. Let $L(\pi)$ denotes its label. Two nodes $\pi_1$ and $\pi_2$ in $\mathit{tree}(\Pi_P)$ are considered compatible if and only if they satisfy two conditions. The first condition is $\mathit{last}(\pi_1) = \mathit{last}(\pi_2)$ where $\mathit{last}(\pi)$ is the last letter in a string $\pi$, i.e., $\pi_1$ and $\pi_2$ must agree on the last abstract state. The second condition is that the future behaviors from $\pi_1$ and $\pi_2$ must be sufficiently similar (i.e., within Angluin's bound~\cite{angluin1988identifying}). Formally, given a node $\pi$ in $\mathit{tree}(\Pi_P)$, we can obtain a probabilistic distribution of the next node by normalizing the labels of the node and its children. In particular, for any event $e \in \Sigma_P$, the probability of going from node $\pi$ to $\pi \cdot \langle e \rangle$, i.e., the probability of observing $e$ after $\pi$, is defined as: $Pr(\pi, \langle e \rangle) = \frac{L(\pi \cdot \langle e \rangle)}{L(\pi)}$. We remark the probability of going from node $\pi$ to itself, i.e., the probability of not observing the next node, is $Pr(\pi,\langle\rangle)=1-\sum_{e\in\Sigma_P}Pr(\pi,\langle e\rangle)$. The multi-step probability from node $\pi$ to $\pi \cdot \pi'$ where $\pi' = \langle e_1, e_2, \cdots, e_k \rangle$, written as $\it{Pr}(\pi, \pi')$, is the product of the one-step probabilities.
\begin{equation}
	\begin{split}
Pr(\pi, \pi') = \it{Pr}(\pi, \langle e_1 \rangle) \times \it{Pr}(\pi \cdot \langle e_1 \rangle, \langle e_2 \rangle)\times\\
	\cdots \times \it{Pr}(\pi \cdot \langle e_1, e_2, \cdots, e_{k-1} \rangle, \langle e_k \rangle)
	\end{split}
\end{equation}
Two nodes $\pi_1$ and $\pi_2$ are compatible if the following is satisfied: for all $\pi \in \Sigma_P^{\star}$,
\begin{equation}
	\begin{split}
|Pr(\pi_1, \pi) - \it{Pr}(\pi_2, \pi)|<\sqrt{6 \epsilon \log(L(\pi_1))/L(\pi_1)} +\\
		\sqrt{6 \epsilon \log(L(\pi_2))/L(\pi_2)}
	\end{split}
	\label{eq:test}
\end{equation}
Intuitively, it means that the distribution of future traces from $\pi_1$ and $\pi_2$ must be similar. We highlight that $\epsilon$ used in Eq.~\ref{eq:test} is a parameter which effectively controls the degree of node merging. Intuitively, a larger $\epsilon$ leads to more node merging and subsequently fewer states in the learned model $\mathcal{D}_P$. It is also required that $\epsilon$ should be larger than 1 to guarantee the convergence of the algorithm.

If $\pi_1$ and $\pi_2$ are compatible, the tree is transformed such that the incoming edge of $\pi_2$ is directed to $\pi_1$. Next, for any $\pi \in \Sigma_P^*$, $L(\pi_1 \cdot \pi)$ is incremented by $L(\pi_2 \cdot \pi)$. The algorithm works by iteratively identifying nodes which are compatible and merging them until there are no more compatible nodes. The order of choosing merging candidates is hierarchical: first in order of tree depth and for a given depth in the alphabet order of the last observation.

Recall that Fig.~\ref{fig:tree0} is the tree representing 100 abstract traces of our running example.
The labels on the nodes are the \emph{numbers} of times the corresponding string is a prefix of some trace in $\Pi_P$. For instance, node $\langle\rangle$ is labeled with 100 because $\langle\rangle$ is a prefix of all traces. This tree can be viewed as the initial learned model which has no generalization. Next, the tree is generalized by merging nodes. Assume that node $\langle 00 \rangle$ and node $\langle 000 \rangle$ in Fig.~\ref{fig:tree0} pass the compatibility test so that they are to be merged. Firstly, we update the node label of $\langle 00 \rangle$ to be the sum of the numbers labeling $\langle 00 \rangle$ and $\langle 000 \rangle$. Secondly, the numbers labeling decedents of $\langle 000 \rangle$ are added to the corresponding decedent nodes of $\langle 00 \rangle$. The result is shown in Fig.~\ref{fig:tree}, where the numbers after the arrow are the ones after node merging. For instance, since $L(\langle 000 \rangle \cdot \langle 1 \rangle)$ is 2, we update the label of node $\langle 001 \rangle$ from 10 to 12. 

After merging all compatible nodes, the last step is to normalize the tree so that it becomes a DTMC $\mathcal{D}_P$. In particular, each node $\pi$ is taken as a state in $\mathcal{D}_P$. The transition probability from $\pi$ to a child $\pi'$ is set to be: $\frac{L(\pi')}{L(\pi)}$ and the probability to itself is set to 1 minus the sum of probabilities to its children accordingly. For instance, in the example shown in Fig.~\ref{fig:tree}, the probability transiting from node $\langle 00 \rangle$ to node $\langle 001 \rangle$ is $\frac{12}{190}$, and the probability to itself is $1-\frac{12}{190}$.

\begin{remark}
Notice that the parameter $\epsilon$ in Eq.~\ref{eq:test} has an immediate effect on the compatible test and thus the learned model. One question would be how to determine the best $\epsilon$ for the best learning outcome. AALERGIA measures how good a learned model $\mathcal{D}_P$ is by the Bayesian Information Criteria (BIC) score. Intuitively, a learned model which generates the input system traces $\Pi_P$ with higher probability and has fewer number of states has a higher BIC score. AALERGIA thus automatically searches for the best $\epsilon$ that has the highest BIC score and output the corresponding model~\cite{AA}.  
\end{remark}

\mypara{GA-based learning algorithm} A genetic-algorithm based learning approach (GA in short hereafter) was proposed in~\cite{wang2017should}. The idea is to reduce the problem of model learning to an optimization problem, i.e., finding an optimal DTMC $\mathcal{D}_P$ that has the highest BIC score. GA encodes a chromosome by randomly merging the nodes in the prefix tree and apply mutations and crossover in the standard way to search for the best model. It has been shown such an approach could outperform AALERGIA sometimes~\cite{wang2017should}. We remark that different model learning approaches can be adopted in this work.

\subsection{Spuriousness Checking}
Given the learned DTMC model $\mathcal{D}_P$ and the property $\mathcal{P}_{\leq r}(\textbf{F}\varphi)$, we then use an existing probabilistic model checker (e.g., PRISM~\cite{kwiatkowska2011prism}) to check whether $\mathcal{D}_P \models \mathcal{P}_{\leq r}(\textbf{F}\varphi)$. There are two cases.

The first case is that $\mathcal{D}_P \models \mathcal{P}_{\leq r}(\textbf{F}\varphi)$ is not satisfied. Since this verification result is based on the learned model $\mathcal{D}_P$, it may not be correct according to $M$. Thus, we must validate this result based on the actual system. To do so, we construct a probabilistic counterexample and check whether the counterexample is spurious. In general, a probabilistic counterexample is in the form of a tree~\cite{cegeneration}. Since we assume that our property is in the form of $\mathcal{P}_{\leq r}(\textbf{F}\varphi)$, we define it to be a set of abstract paths of $\mathcal{D}_P$ for simplicity.\\
\begin{definition}
Given a DTMC $\mathcal{D}_P$ and a property $\mathcal{P}_{\leq r}(\textbf{F}\varphi)$, a counterexample is a set $C \subseteq \mathit{Path}_{fin}(\mathcal{D}_P)$ such that $\mathit{last}(\pi) \models \varphi$ for every path $\pi$ in $C$, and $\sum_{\pi \in C}\mathcal{P}(\pi, \mathcal{D}_P)>r$.
\end{definition}

\noindent Intuitively, a probabilistic counterexample for $\mathcal{P}_{\leq r}(\textbf{F}\varphi)$ is a set of finite paths in $\mathcal{D}_P$ whose accumulated probability is larger than $r$. There are existing approaches to construct such counterexamples~\cite{cegeneration}. In our setting, because $\mathcal{D}_P$ is learned based on a set of abstract traces, a probabilistic counterexample may be spurious. A probabilistic counterexample $C$ is spurious if and only if the probability measure of the concrete paths of $M$ corresponding to $C$ is less than $r$. Given an abstract path $\pi=\langle a_1, a_2, \cdots, a_n \rangle \in Path_{fin}(\mathcal{D}_P)$, we write $\gamma(\pi)$ to denote the set of concrete paths $\{\langle s_1, s_2, \cdots, s_n \rangle \in Path_{fin}(M) | \forall i.~\alpha_P(s_i) = a_i\}$ which become $\pi$ after predicate abstraction using $P$. We write $\gamma(C)$ to denote the set of corresponding concrete paths: $\{ \pi' | \exists \pi \in C,\ \pi' \in \gamma(\pi)\}$.\\
\begin{definition}
Given the system $M$, a DTMC $\mathcal{D}_P$, a property $\mathcal{P}_{\leq r}(\textbf{F}\varphi)$ and a counterexample $C$ in terms of $\mathcal{D}_P$ and $\mathcal{P}_{\leq r}(\textbf{F}\varphi)$, $C$ is a spurious counterexample of $M$ if and only if $\sum_{\pi \in \gamma(C)} \mathcal{P}(\pi, M) \leq r$.
\end{definition}

\noindent Recall that $M$ can be viewed as a DTMC and $\mathcal{P}(\pi, M)$ is the probability of path $\pi$ in $M$. The above notion of spurious counterexample is a special case of that of probabilistic automaton defined in~\cite{probcegar}. Checking whether a probabilistic counterexample $C$ is spurious or not in our setting is however more challenging than that in the setting of~\cite{probcegar}. The reason is that we do not have a model of $M$ and thus the probability of a concrete path in $M$ cannot be calculated. Our solution is to adopt hypothesis testing~\cite{YounesThesis} to test whether the hypothesis $\sum_{\pi \in \gamma(C)}\mathcal{P}(\pi, M) > r$ holds given certain error bounds. If it does, we report the counterexample to the user together with the error bounds. Otherwise, we conclude that it is a spurious counterexample and proceed to the next step, i.e., abstraction refinement. 

In the following, we briefly introduce how hypothesis testing is used for spuriousness checking. 
Hypothesis testing is a statistical process to decide the truthfulness of two mutual exclusive statements. One is $H_0$: the null hypothesis that the counterexample is not spurious, i.e., $\sum_{\pi \in \gamma(C)}\mathcal{P}(\pi, M) > r$. The other is $H_1$: the alternative hypothesis that the counterexample is spurious, i.e., $\sum_{\pi \in \gamma(C)}\mathcal{P}(\pi, M) \leq r$. The probability of making an error is bounded by $(\alpha,\beta)$, such that the probability of a Type-\uppercase\expandafter{\romannumeral1} (respectively, a Type-\uppercase\expandafter{\romannumeral2}) error, which rejects $H_0$ (respectively, $H_1$) while $H_0$ (respectively, $H_1$) holds, is less or equal to $\alpha$ (respectively, $\beta$). The test needs to be relaxed with a confidence interval $(r-\delta,r+\delta)$, where neither hypothesis is rejected and the test continues to bound both types of errors~\cite{smcover}. In practice, the parameters (i.e., $(\alpha, \beta)$, and $\delta$) can often be decided by how much testing resources are available. In general, more resource is required for a smaller error bound.


Hypothesis testing works by keeping sampling traces from $M$ until a stopping condition is satisfied. There are two main methods to decide when sampling can be stopped. One is fixed-size sampling test, which requires running all the predefined number of tests. One difficulty of this approach is to find the appropriate number of tests to be performed such that the error bounds are valid\cite{smcover}. 
In this work, we adopt the alternative approach, which is sequential probability ratio test (SPRT). SPRT dynamically decides whether to reject or not a hypothesis every time we obtain a new trace, which yields a variable sample size. SPRT is usually faster than fixed-size sampling test as the testing process ends as soon as a conclusion is made. In the following, we briefly introduce how SPRT works and readers can refer to~\cite{smcover} for details. 

The basic idea of SPRT is to calculate the probability ratio each time after observing a trace and evaluate two stopping conditions~\cite{wald}. If either of the conditions is satisfied, the testing stops and returns which hypothesis is rejected. Let $n$ be the number of sampled traces so far, the decision is based on the following probability ratio:
\begin{equation}
\frac{p_{1n}}{p_{0n}}=\prod_{i=1}^n
\frac{P(B_i=b_i|p=p_1)}{P(B_i=b_i|p=p_0)}=
\frac{p_1^{d_n}(1-p_1)^{n-d_n}}{p_0^{d_n}(1-p_0)^{n-d_n}}
\end{equation}
, where $p_1=r+\delta$, $p_0=r-\delta$, $B_i$ is a random variable with boolean values representing whether a new sampled trace is in the counterexample (1 if yes and 0 otherwise), $d_n$ is the number of sampled trace which are in the counterexample. We will accept $H_0$ (reject $H_1$) if $\frac{p_{1n}}{p_{0n}}\ge \frac{1-\beta}{\alpha}$ and $H_1$ (reject $H_0$) if $\frac{p_{1n}}{p_{0n}}\le \frac{\beta}{1-\alpha}$. Otherwise, SPRT continues sampling until either $H_0$ or $H_1$ is rejected. 
We remark that SPRT is guaranteed to terminate with probability 1~\cite{wald}.


In our running example, hypothesis testing is done by repeatedly generating messages from some user in the network and observe the resultant traces.
Assume there are now a total of $n$ traces in $\Pi$ (including the ones obtained initially) and $b$ of them are in $\gamma(C)$. We can then calculate whether to accept $H_0$ or $H_1$ based on $n$ and $b$ in the standard way. If neither can be accepted, we continue sampling and updating $n$ and $b$ until one of the hypothesis is accepted. The following theorem follows the correctness of hypothesis testing immediately. \\
\newtheorem{theorem}{Theorem}
\begin{theorem}
If LAR reports a counterexample, the probability of $M$ satisfying the property is `bounded' by $\beta$. \hfill $\Box$
\end{theorem}
\newtheorem{proof}{Proof}
\begin{proof}
 Notice that if the counterexample is not spurious, then the property is violated by the system. So the probability that $M$ satisfies the property is no larger than the probability that the counterexample is actually spurious while a counterexample is reported, which is the Type-\uppercase\expandafter{\romannumeral2} error bounded by $\beta$.
\end{proof}
\noindent In the case that the model checker confirms that $\mathcal{D}_P \models \mathcal{P}_{\leq r}(\textbf{F}\varphi)$, it is guaranteed that the the original system $M$ satisfies the property \emph{if an unbounded number of traces are used for learning.} \\
\begin{theorem}
If $\Pi$ has an unbounded number of traces, $\mathcal{D}_P \models \mathcal{P}_{\leq r}(\textbf{F}\varphi)$ with probability 1 only if $M$ satisfies the property. \hfill $\Box$
\end{theorem}

\begin{proof}
The proof of the theorem is as follows. First, the AALERGIA algorithm eventually converges to the actual underlying DTMC~\cite{AA,AAJ} with probability 1, which is a quotient automaton of the original model $M$ (with respect to the predicate abstraction). Second, it is established in~\cite{probcegar} that the quotient automaton simulates the original system $M$. Thus, it is safe to conclude that $M$ satisfies the property if its quotient automaton does.
\end{proof}

However, the number of traces in $\Pi$ is often limited in practice and there is no guarantee that the learned model $\mathcal{D}_P$ has converged to a model which simulates $M$. Our remedy to this problem is that we present the model $\mathcal{D}_P$ as a part of the verification result. That is, we report that the property is satisfied by the system, provided that $\mathcal{D}_P$ is a correct model of the system. We remark that because there is no precise system model, our goal is not to completely verify the system. Rather, our goal is to provide evidence (a counterexample or a model) with confidence on why we believe that the verification result is sound. If we would like to have certain confidence on the ``correctness'' of the learned model $\mathcal{D}_P$, one way is to make sure that the number of traces in $\Pi$ satisfies certain constraints. We skip the details as this is beyond the scope of this work. A discussion on the number of samples required to guarantee that the difference between the transition probabilities of $\mathcal{D}_P$ and $M$ is given in Appendix A.


\subsection{Refinement} \label{refinement}

If the counterexample $C$ is spurious after hypothesis testing, the abstract model $\mathcal{D}_P$ is to be refined so as to rule out the spurious counterexample. Following the idea of CEGAR~\cite{clarke2003counterexample,probcegar}, this is achieved by identifying a new predicate for predicate abstraction. Due to the lack of a system model, standard methods for generating new predicates (e.g., weakest pre-condition calculation~\cite{clarke2003counterexample} or interpolation~\cite{henzinger2004abstractions}) are infeasible in our setting. We thus adopt techniques from the machine learning community to solve the problem. 

\begin{figure}[t]
\centering
\scalebox{.8}{
\begin{tikzpicture}
	\tikzset{cstate/.style={circle,draw,thin,minimum size=0.001cm}}
	\node [state,minimum size=1.6cm] (s1) at (0,0)  {};
	\node [state,minimum size=1.6cm] (s2) at (3,0)  {};
	\node [state,minimum size=1.6cm] (s3) at (6,0)  {};
	
	\node[circle,fill,inner sep=2pt](a1) at (0.4,0.2) {};
	\node[circle,fill,inner sep=2pt](a2) at (0.25,-0.35) {};
	\node[circle,fill,inner sep=2pt](a3) at (0,0.5) {};
	\node[circle,fill,inner sep=2pt](a4) at (-0.2,-0.4) {};
	\node[circle,fill,inner sep=2pt](a5) at (-0.5, 0) {};
	
	\node[circle,fill,inner sep=2pt](b1) at (3.4,0.2) {};
	\node[circle,fill,inner sep=2pt](b2) at (3.25,-0.35) {};
	\node[circle,fill,inner sep=2pt](b3) at (3,0.5) {};
	\node[circle,fill,inner sep=2pt](b4) at (2.8,-0.4) {};
	\node[circle,fill,inner sep=2pt](b5) at (2.5, 0) {};
	
	\node[circle,fill,inner sep=2pt](c1) at (6.4,0.2) {};
	\node[circle,fill,inner sep=2pt](c2) at (6.25,-0.35) {};
	\node[circle,fill,inner sep=2pt](c3) at (6,0.5) {};
	\node[circle,fill,inner sep=2pt](c4) at (5.8,-0.4) {};
	\node[circle,fill,inner sep=2pt](c5) at (5.5, 0) {};
	
	\draw[->] (a1) -- (b3);
	\draw[->] (a1) -- (b5);
	\draw[->] (a2) -- (b4);
	\draw[->] (b3) -- (c3);
	\draw[->] (b1) -- (c5);
	\draw[->] (b2) -- (c4);

\end{tikzpicture}}
\caption{Illustration of broken paths}
\label{fig:split}
\end{figure}
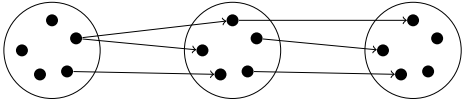

\mypara{Identify spurious transitions} The first step of identifying a new predicate is to identify spurious transitions in the counterexample. Fundamentally, spuriousness is the result of an inappropriate level of abstraction. Fig.~\ref{fig:split} shows how abstraction could induce more probability measure over abstract paths, due to broken paths~\cite{kroening2004counterexample}. The probability of paths in the counterexample $C$ may be inflated. Intuitively, if a probabilistic counterexample $C$ is spurious, there must be at least one path $\pi$ in $C$ whose probability is inflated, i.e., $\pi$ has certain probability $d$ in $\mathcal{D}_P$, whereas the accumulated probability of $\gamma(\pi)$ in $M$ is less than $d$. Furthermore, there must be at least one transition in $\pi$ whose probability in $\mathcal{D}_P$ is higher than the accumulated probability of the corresponding concrete transitions in $M$. The idea is then: if we are able to identify such an abstract path $\pi$ and such a transition in $\pi$, we could refine $\mathcal{D}_P$ such that the transition is no longer associated with the inflated probability so that $C$ may no longer be a counterexample. In the following, we define such transitions and later use them to find new predicates.

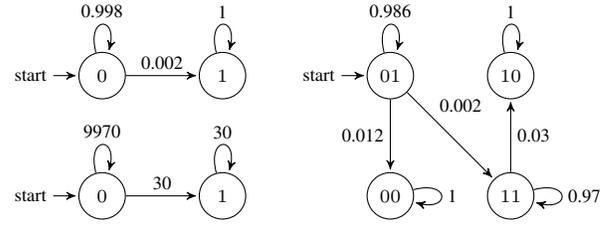
\begin{figure}[t]
	\centering
{\scriptsize \begin{tikzpicture}[>=stealth',shorten >=1pt,auto,node distance=1.6cm]
  \node[initial,state] (A0)      {$0$};
  \node[state]         (A1) [right of=A0]  {$1$};
  
  \node[initial,state] (B0) [below of=A0]     {$0$};
  \node[state]         (B1) [right of=B0] {$1$};

    \node[initial,state] (N0)  [right = of A1]    {$01$};
    \node[state]         (N1) [below of=N0]  {$00$};

    \node[state] (N2) [right of=N0]     {$10$};
    \node[state]         (N3) [right of=N1] {$11$};
	
  \path[->] (A0)  edge [loop above] node {0.998} (A0)
             edge              node {0.002} (A1)
			 
			(A1)  edge [loop above] node {1} (A1)
			            
        	(B0) edge [loop above]  node {9970} (B0)
             edge              node {30} (B1)
			 
			(B1)  edge [loop above] node {30} (B1)

			(N0)  edge [loop above] node {0.986} (N0)
				               edge              node [left] {0.012} (N1)
				  			 edge			   node [pos=0.3]{0.002} (N3)
			 
				  			(N1)  edge [loop right] node {1} (N1)
			            
				          	(N2) edge [loop above]  node {1} (N2)

				  			(N3)  edge [loop right] node {0.97} (N3)
				  				  edge 			node [right] {0.03} (N2);		           
\end{tikzpicture}}
\caption{Identify a most spurious transition}
\label{fig:spurious}
\end{figure}

Formally, let $C = \{\pi_1, \pi_2, \cdots, \pi_n\}$ be a spurious probabilistic counterexample. Let $\langle \mathrm{s}, \mathrm{s}'\rangle$ be a pair of consecutive states which form a transition of a path $\pi_i$ in $C$. Let $\gamma(\mathrm{s})$ denote the set of concrete states in $M$ which become $\mathrm{s}$ after abstraction, i.e., $\gamma(\mathrm{s}) = \{x \in \Sigma | \db{x}_P = \mathrm{s}\}$, and $|\gamma(\mathrm{s})|$ denote the number of concrete states in $\gamma(\mathrm{s})$.
We write ${\mathcal{P}}(\langle \mathrm{s}, \mathrm{s}'\rangle, M)$ to be $\sum_{s_0 \in \gamma(\mathrm{s}), s_0' \in \gamma(\mathrm{s'})}\#(\langle s_0, s_0' \rangle, M)/|\gamma(\mathrm{s})|$, i.e., the probability of having a corresponding concrete transition in $M$. $\langle \mathrm{s}, \mathrm{s'} \rangle$ is called a spurious transition if $\mathcal{P}(\langle \mathrm{s}, \mathrm{s'} \rangle, \mathcal{D}_P) > {\mathcal{P}}(\langle \mathrm{s}, \mathrm{s'}\rangle, M)$. Since we do not have the model $M$, it is impossible to compute ${\mathcal{P}}(\langle \mathrm{s}, \mathrm{s'}\rangle, M)$. Rather, we estimate it based on $\Pi$ (which now contains all traces sampled during the hypothesis testing in the previous step). That is, given the set of sample traces $\Pi$, we estimate ${\mathcal{P}}(\langle \mathrm{s}, \mathrm{s'}\rangle, M)$ by $\#\langle \mathrm{s}, \mathrm{s'} \rangle/\#\mathrm{s}$ where $\#\langle \mathrm{s}, \mathrm{s'}\rangle$ is the number of times the transitions take place and $\#\mathrm{s}$ is the number of times $\mathrm{s}$ is visited. Next, we identify the most spurious transition for refinement, i.e., the transition $\langle \mathrm{s}, \mathrm{s'}\rangle$ s.t.~$\mathcal{P}(\langle \mathrm{s}, \mathrm{s'} \rangle, \mathcal{D}_P) - \#\langle \mathrm{s}, \mathrm{s'}\rangle/\#\mathrm{s}$ is the largest of all transitions.

Algorithm~\ref{alg:idenst} gives a detailed description on how we identify the most spurious transitions. We first collect all the transitions in the learned model $\mathcal{D}_P$ at line 1. Then we iterate through all the transitions at line 4. For each transition, we apply Algorithm~\ref{alg:ct} to obtain $\#\mathrm{s}$ and $\#\langle \mathrm{s}, \mathrm{s'}\rangle$. The basic idea is to check all the traces in $\Pi$ and count how many times we observe an abstract state $\mathrm{s}$ and a transition $\langle \mathrm{s}, \mathrm{s'} \rangle$. We then calculate the difference between $\mathcal{P}(\langle \mathrm{s}, \mathrm{s'} \rangle, \mathcal{D}_P) - \#\langle \mathrm{s}, \mathrm{s'}\rangle/\#\mathrm{s}$ at line 7. After iterating through all the transitions, we sort them in descending order according to the difference at line 11 and return the sorted transitions at line 12.  

Fig.~\ref{fig:spurious} shows the first learned model $\mathcal{D}_P$ of our running example on the top-left, where there are two states 0 and 1 representing the state $observe0 \leq 1$ and $observe0 > 1$ respectively. There are in total three transitions: $\langle 0, 0\rangle$ , $\langle 0, 1\rangle$, $\langle 1, 1\rangle$. Based on the learned model, $\mathcal{P}(\langle 0, 0 \rangle, \mathcal{D}_P)$ is 0.998; $\mathcal{P}(\langle 0, 1 \rangle, \mathcal{D}_P)$ is 0.002 and $\mathcal{P}(\langle 1, 1 \rangle, \mathcal{D}_P)$ is 1. Assume that the estimation of ${\mathcal{P}}(\langle 0, 0\rangle, M)$ is 0.997; the estimation of ${\mathcal{P}}(\langle 0, 1\rangle, M)$ is 0.003; and the estimation of ${\mathcal{P}}(\langle 1, 1\rangle, M)$ is 1. By calculating the difference between the transition in the learned model and its estimation, we find that the transition $\langle 0,0 \rangle$ has the largest difference (i.e., the most spurious), which is used for abstraction refinement.

\begin{algorithm}[t]
    Collect all the transitions $\textit{Trans}$ in $\mathcal{D}_P$\;
    Let $P_{\textit{Diffs}}$ be the array storing the probability differences\; 
    Let $\textit{ATrans}$ be the set of analyzed transitions and initialized to be empty\;
    \While {$\textit{Trans}!=\emptyset$} {
        Randomly pick a transition $\textit{Tran}=\langle \mathrm{s},\mathrm{s'}\rangle$ from $\textit{Trans}$\;
        Apply Algorithm~\ref{alg:ct} $\textit{CountTran}(\textit{Tran},\Pi,P)$ to get $\# \mathrm{s}$ and $\# \langle \mathrm{s},\mathrm{s'} \rangle$ \;
        Calculate $P_{\textit{diff}}$ as the difference between $\# \langle \mathrm{s},\mathrm{s'} \rangle/\# \mathrm{s}$ and the transition probability of $\langle \mathrm{s},\mathrm{s'}\rangle$ in $\mathcal{D}_P$\;
        Add $P_{\textit{diff}}$ to $P_{\textit{Diffs}}$\;  
        Add $\textit{Tran}$ to $\textit{ATrans}$\;
        Remove $\textit{Tran}$ from $\textit{Trans}$\;
    }
    Sort $\textit{ATrans}$ according to $P_{\textit{Diffs}}$ in descending order\;
    \Return $\textit{ATrans}$
\caption{$\textit{IdentifyST}(\Pi,\mathcal{D}_P,P)$}
\label{alg:idenst}
\end{algorithm}

\begin{remark}
We remark that the identified spurious transitions may be different if two different counterexamples are provided. Note that we follow the approach in~\cite{cegeneration}, which calculates a smallest counterexample, i.e., a minimum set of paths to form a counterexample.
\end{remark}

\mypara{Labeled data collection} Next, we identify a new predicate based on the most spurious transition $\langle \mathrm{s}, \mathrm{s'} \rangle$. Intuitively, the reason that the transition is spurious is that many abstract paths in $\mathcal{D}_P$ going through this transition are infeasible. This is illustrated in Fig.~\ref{fig:split}, where a number of paths are `broken' at the middle state. Note that each abstract state in $\mathcal{D}_P$ groups a set of concrete states in $\Sigma$. Intuitively, $\langle \mathrm{s}, \mathrm{s'}\rangle$ is spurious because $\mathrm{s}$ groups some states which cannot transit to any concrete states in $\mathrm{s'}$, with states which can. As a result, the probability of transiting from $\mathrm{s}$ to $\mathrm{s'}$ is inflated. Thus, in order to prune the spurious counterexample, we need to refine the model such that $\mathrm{s}$ is split in a way such that the states which cannot transit to any concrete states in $\mathrm{s'}$ are separated from the rest. In the following, we first collect these two sets of states and identify a predicate for separating them using a classification algorithm.

\begin{algorithm}[t]


\For{every path $\pi\in\Pi$}{
	Abstract $\pi$ using $P$ to get the abstract path $\pi_P$\;
	Let $l$ be the length of $\pi$ \;
	\For{$s_i \in \pi_P[2\cdots l]$}{
		\If{$\pi_P[i-1]=\mathrm{s}$}{
			$\# \mathrm{s}=\# \mathrm{s}+1$\;
			\If{$\mathcal{D}_P(\pi_P[i-1],\pi_P[i])=\mathrm{s'}$}{
				$\# \mathrm{s'}=\# \mathrm{s'}+1$\;
				Add the corresponding concrete state of $\pi_P[i-1]$ to $\gamma^+(\mathrm{s}, \mathcal{D}_P, \Pi)$\;
			}
			\Else{
				Add the corresponding concrete state of $\pi_P[i-1]$ to $\gamma^-(\mathrm{s}, \mathcal{D}_P, \Pi)$\;
			}
		}
	}
}
\Return $\# \mathrm{s}$, $\# \langle \mathrm{s},\mathrm{s'}\rangle$, $\gamma^+(\mathrm{s}, \mathcal{D}_P, \Pi)$, $\gamma^-(\mathrm{s}, \mathcal{D}_P, \Pi)$
\caption{$\textit{CountTran}(\langle \mathrm{s},\mathrm{s'}\rangle,\Pi, P)$}
\label{alg:ct}
\end{algorithm}

Given any concrete execution $\pi = \langle s_1, s_2, \cdots, s_n \rangle$ in $\Pi$, it is clear that we can map any state $s_i$ in the sequence to an abstract state in $\mathcal{D}_P$. Let $\gamma(\mathrm{s}, \mathcal{D}_P, \Pi)$ be the set of concrete states in any execution in $\Pi$ which are mapped to state $\mathrm{s}$ in $\mathcal{D}_P$. We define $\gamma^+(\mathrm{s}, \mathcal{D}_P, \Pi) = \{ x | x \in \gamma(\mathrm{s}, \mathcal{D}_P, \Pi) \land \exists \langle \cdots, x, y, \cdots \rangle \in \Pi \land y \in \gamma(\mathrm{s'}, \mathcal{D}_P, \Pi)\}$, i.e., the set of concrete states in $\mathrm{s}$ which do transit to a concrete state in $\mathrm{s'}$ in one of the concrete trace. We define $\gamma^-(\mathrm{s}, \mathcal{D}_P, \Pi)$ to be $\gamma(\mathrm{s}, \mathcal{D}_P, \Pi) - \gamma^+(\mathrm{s}, \mathcal{D}_P, \Pi)$. 
Algorithm~\ref{alg:ct} shows how we collect the positive and negative data from system traces for a given spurious transition. For every trace in $\Pi$, we first abstract $\pi$ with $P$ at line 2. Starting from the second state in $\pi_P$ to the end, we iteratively check whether its previous state is $\mathrm{s}$ at line 5. If it is true, we then check whether the current state is $\mathrm{s'}$ at line 7. If yes, we add the current state to $\gamma^+(\mathrm{s}, \mathcal{D}_P, \Pi)$ at line 9; otherwise, we add it to $\gamma^-(\mathrm{s}, \mathcal{D}_P, \Pi)$ instead. Fig.~\ref{fig:collect} illustrates how the above works with our running example. Assume that at the top is a concrete trace and at the bottom is the corresponding abstract trace. Given that we identified previously that the spurious transition is $\langle 0,0\rangle$, the concrete states $s_1$, $s_2$ and $s_3$ are collected to the set $\gamma^+(\mathrm{s}, \mathcal{D}_P, \Pi)$ and state $s_4$ is collected into $\gamma^-(\mathrm{s}, \mathcal{D}_P, \Pi)$.

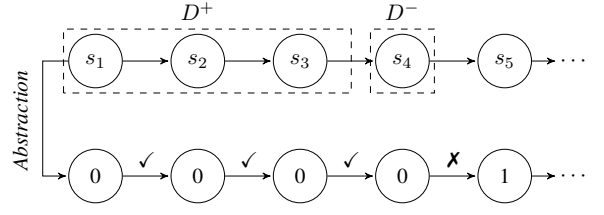
\begin{figure}[]
\centering
\scalebox{.85}{
\begin{tikzpicture}[>=stealth',shorten >=1pt,auto,every node/.style={minimum size=1.5em,inner sep=1}]
	\node [state] (s1) at (0,0)  {$s_1$};
	\node [state] (s2) at (1.6,0) {$s_2$};
	\node [state] (s3) at (3.2,0) {$s_3$};
	\node [state] (s4) at (4.8,0) {$s_4$};
	\node [state] (s5) at (6.4,0) {$s_5$};
	\node (s6) at (7.5,0) {$\cdots$};
	\node (l1) [above=0.1cm of s2] {$D^{+}$};
	\node (l1) [above=0.1cm of s4] {$D^{-}$};
	
	\node [state] (s11) at (0,-1.8)  {0};
	\node [state] (s21) at (1.6,-1.8) {0};
	\node [state] (s31) at (3.2,-1.8) {0};
	\node [state] (s41) at (4.8,-1.8) {0};
	\node [state] (s51) at (6.4,-1.8) {1};
	\node (s61) at (7.5,-1.8) {$\cdots$};
	
	\draw[->] (s1) -- (s2);
	\draw[->] (s2) -- (s3);
	\draw[->] (s3) -- (s4);
	\draw[->] (s4) -- (s5);
	\draw[->] (s5) -- (s6);
	
	\draw[->] (s11) -- node [above] {$\checkmark$} (s21);
	\draw[->] (s21) -- node [above] {$\checkmark$} (s31);
	\draw[->] (s31) -- node [above] {$\checkmark$} (s41);
	\draw[->] (s41) -- node [above] {$\xmark$} (s51);
	\draw[->] (s51) -- (s61);
	
	\draw[->] (s1.west) -- ++(-0.4cm,0) 
	  |- node [left,xshift=-0.35cm,yshift=1.7cm,rotate=90] {\textit{Abstraction}} (s11.west);
	  
	\draw[dashed] (-0.5,-0.5) rectangle (4,0.5);
	\draw[dashed] (4.3,-0.5) rectangle (5.3,0.5);

\end{tikzpicture}}
\caption{Labeled data collection at spurious transition $\langle 0,0 \rangle$}
\label{fig:collect}
\end{figure}

\mypara{Generate a new predicate} Once we obtain the labeled data of the spurious transition, we then adopt a supervised classification technique to generate a new predicate for refinement. In particular, we use Support Vector Machine (SVM)~\cite{cortes1995support,chang2011libsvm} as it is reasonably scalable and produces a classifier in the form of a predicate. SVM is a supervised machine learning algorithm for classification and regression analysis. In this work, we use its binary classification functionality. Mathematically speaking, the binary classification functionality of (linear) SVM works as follows. Given the two sets of concrete states $\gamma^+(\mathrm{s}, \mathcal{D}_P, \Pi)$ and $\gamma^-(\mathrm{s}, \mathcal{D}_P, \Pi)$, SVM generates, if there is any, a linear classifier (hyperplane) in the form of $e = \sum_i c_i*x_i\geq c$ where $x_i \in V$ is a variable value and $c_i$ and $c$ are learned coefficients such that (1) for every state $x \in \gamma^+(\mathrm{s}, \mathcal{D}_P, \Pi)$, $\db{e}_x = 1$ and (2) for every state $x \in \gamma^-(\mathrm{s}, \mathcal{D}_P, \Pi)$, $\db{e}_x = 0$. As a simple example, for a classification task with only two features (variables), a linear classifier (hyperplane) is a line that linearly separates the two sets of data while providing a maximum margin between the two sets, which is shown in Fig.~\ref{fig:svm}.

\begin{figure}
\centering
\includegraphics[width=.35\textwidth]{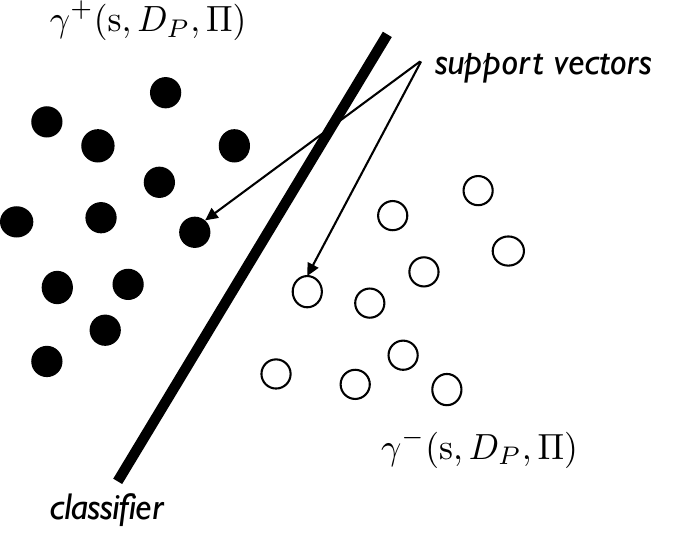}
\caption{A simple example of SVM with only two variables. The filled circles are the collected positive states and the empty circles are collected negative states.}
\label{fig:svm}
\end{figure}


In our running example, since transition $\langle 0, 0\rangle$ is the most spurious, we collect the two sets of concrete states in order to split state 0 in the DTMC shown on the left of Fig.~\ref{fig:spurious}.
Recall that a concrete state contains the valuation of 32 variables. We feed the two sets into SVM to generate a classifier. Notice that a boolean variable is mapped to 1 if its value is true and 0 otherwise. 
The result is a predicate constituted by all 32 variables. To simplify the predicate, we start from the variable with the largest coefficient and try to identify a classifier with a minimum number of variables~\cite{cortes1995support}, e.g., a predicate constituted by a minimum number of variables and yet is able to classify the two sets. The result is the predicate $new<runCount$. The predicate is then added into $P$, and used to refine the abstraction. The right part of Fig.~\ref{fig:spurious} shows the learned model after refinement, where the second bit of the state label captures the value of the new predicate. It can be observed that a new terminal state $00$ is introduced, which effectively reduces the probability of reaching $observe0> 1$. 

\begin{remark}
We remark that $\gamma^+(\mathrm{s}, \mathcal{D}_P, \Pi)$ and $\gamma^-(\mathrm{s}, \mathcal{D}_P, \Pi)$ may not be linearly separable, in which case linear SVM could not generate a satisfactory linear classifier (i.e., the obtained classifier has low classification accuracy). Although SVM could adopt different kinds of kernels (e.g., polynomial and RBF) to identify a polynomial or even more expressive classifier, we prefer linear classifiers as the learned classifier and model are presented as a part of the verification results and the easier it is, the easier it is for human understanding. In case a linear classifier for separating $\gamma^+(\mathrm{s}, \mathcal{D}_P, \Pi)$ and $\gamma^-(\mathrm{s}, \mathcal{D}_P, \Pi)$ does not exist, we move on to identify a predicate using the same approach based on the next most spurious transition.
\end{remark}

\subsection{Overall Algorithm}

\begin{algorithm}[t]
    let $P$ be the set of atomic propositions in $\varphi$\;
    \While {true} {
        construct abstract executions $\Pi_P$ based on $\Pi$ and $P$\;
        apply a model learning algorithm (e.g. AALERGIA with $\epsilon$ ranging from 1 to $\epsilon_{max}$) to learn a model $\mathcal{D}_P$ based on $\Pi_P$\;
        probabilistic model check $\mathcal{D}_P$ against $\varphi$\;
        \If {$\mathcal{D}_P \models \varphi$} {
            report $\varphi$ is verified, model $\mathcal{D}_P$\;
            \Return\;
        }
        construct a probabilistic counterexample $C$ using the approach in~\cite{cegeneration}\;
        run hypothesis testing on $C$ with error bounds $(\alpha,\beta)$ and indifference region $\delta$\;
        \If {$C$ is not spurious} {
            report $\varphi$ is violated and the probability of error is bounded by $(\alpha,\beta)$\;
            \Return\;
        }
        identify the most spurious transitions $\langle \mathrm{s}, \mathrm{s'}\rangle$ in $C$\;
        collect $\gamma^+(\mathrm{s}, \mathcal{D}_P, \Pi)$ and $\gamma^-(\mathrm{s}, \mathcal{D}_P, \Pi)$\;
        apply SVM to identify a predicate $p$ separating the two sets\;
        add $p$ into $P$\;
    }
\caption{$LAR(\Pi, \mathcal{P}_{\leq r}(\textbf{F} \varphi), \epsilon_{max}, \alpha, \beta, \delta)$}
\label{alg:polynomialSVM}
\end{algorithm}

The overall algorithm is shown as Algorithm~\ref{alg:polynomialSVM}. The inputs of the algorithms are a set of concrete executions $\Pi$, a property in the form of $\mathcal{P}_{\leq r}(\textbf{F} \varphi)$, and the parameters $(\alpha,\beta)$ and $\delta$ for hypothesis testing. During each iteration of the loop from line 2 to 17, we start with constructing a set of abstract traces based on $\Pi$ and a set of predicates $P$. Note that we set the initial set of predicates for abstraction to be the set of propositions in the property. Next, an abstract DTMC $\mathcal{D}_P$ is learned using a model learning algorithm. We then verify $\mathcal{D}_P$ against the property using PMC. If the property is verified, we conclude that the system is verified and present $\mathcal{D}_P$ as a part of the evidence. Otherwise, we construct a probabilistic counterexample $C$ (i.e., a set of abstract paths) at line 9. The spuriousness of $C$ is then checked through hypothesis testing at line 10. If it passes the test, it is returned as a counterexample. In addition, we output the bounded probability of $C$ being a spurious counterexample. Otherwise, at line 14, we identify the most spurious transition and obtain a new predicate at line 16. After adding the new predicate into $P$, we restart the process from line 2.

In our running example, after adding the new predicate, we learn the model as shown on the right of Fig.~\ref{fig:spurious}. Note that state 0 in the model on the left is split into two states 00 and 01. This is due to the new predicate. Note also that the spurious transition $\langle 0, 0 \rangle$ is split into two, which effectively reduces the probability of the counterexample. Verifying this new model against the property using PMC returns true and thus we successfully verify the system without requiring a system model as input.
\begin{remark}
It is hard to analyze the complexity of LAR as it depends on several factors, i.e., the initial set of predicates $P$, the complexity of the model learning algorithm, the complexity of hypothesis testing, the complexity of SVM, and whether there is a linear classifier, etc. We thus rely on empirical studies to evaluate the efficiency and effectiveness of the algorithm in the next section. 
Note that if SVM fails to find a classifier for all the spurious transitions, Algorithm~\ref{alg:polynomialSVM} terminates and reports the verification is unsuccessful. Otherwise, it either reports true with a supporting model or a counterexample.
\end{remark}


\section{Evaluation} \label{evaluation}
Our approach has been implemented as a self-contained toolkit named LAR (available at GitHub~\cite{tool})
with about 6K lines of Java code. LAR relies on the LIBSVM library~\cite{chang2011libsvm,abeel2009java} for generating new predicates and PRISM model checker for PMC~\cite{kwiatkowska2011prism}. All models and detailed results discussed below are available at~\cite{results}. 

\subsubsection*{Experiment settings}

We identify and compare our approach with two latest state-of-the-art approaches developed based on stochastic regular grammatical inference. One is the AALERGIA algorithm (hereafter AA)~\cite{AA,AAJ}. The other is the GA-based approach (hereafter GA) in~\cite{wang2017should}.


Our test subjects include DTMC models from the PRISM benchmark suite~\cite{KNP12b}, as well as a real-world water treatment system (SWaT)~\cite{swat}. 
Table~\ref{tb:cases} summarizes the system configuration and the property to verify of the test subjects, where the detailed description of the PRISM benchmark systems can be found at~\cite{prism_models}. SWaT~\cite{swat} is a complicated system which involves a series of water treatment processes like ultra-filtration, chemical dosing, dechlorination through an ultraviolet system, etc. We regard SWaT as a representative of our target complex systems. 
We conduct our experiments on a simulator of the SWaT due to safety concerns. The simulator contains an exact Python translation (about 3K LOC) of the control software in SWaT and a set of Ordinary Differential Equations (ODEs) for simulating the environment. The property we verify is how likely the raw water tank in the system would go underflow. For each model, a set of system traces are first obtained through random simulation. For a fair comparison, we use the same set of system traces to learn from for all the algorithms. The total length of the traces for all the systems is set to be 20000. 
To validate the verification results, we need to know the actual results (last column of Table~\ref{tb:cases}). For each benchmark in the PRISM benchmark suite~\cite{KNP12b}, we apply PRISM on their corresponding PRISM model to obtain the actual probability. Since we do not have the actual model of SWaT, we estimate the verification result based on a large number of traces (e.g., one week's simulation). For each model, LAR is set to verify a property of the form $\mathcal{P}_{\leq r}(\textbf{F} \varphi)$. In our experiment, $r$ is set to be 20\% above the actual $\mathcal{P}(\textbf{F} \varphi)$. 

\begin{table}[t]
	\centering
	\caption{Case studies under evaluation. Column `Config' is the system configuration; column `Property' is the property to verify; column `\#states' is the state space of the original system; column `$P_{act}$' is the actual probability of the property being satisfied.} 
	\begin{tabular}{cccccc}
	\toprule
	Case study 	& \multirow{2}{*}{Config} &\multirow{2}{*}{Property} &\multirow{2}{*}{\#states} &\multirow{2}{*}{$P_{act}$} \\
	(parameters) & & & & \\
	\midrule
	\multirow{4}{*}{\specialcell{$crowds$\\(R,S)}} & 5,5 & \multirow{4}{*}{Positive} & 8k & 0.1458 \\
	 & 5,10 & & 111k & 0.1048 \\
	 & 5,15 & & 592k & 0.0922 \\
	 & 5,20 & & 2062k  & 0.0861 \\
	\midrule
	\multirow{4}{*}{\specialcell{$egl$\\(L,N)}} & \multirow{2}{*}{2,5} & Unfair for A & \multirow{2}{*}{29k} & 0.5156 \\
	& & Unfair for B & & 0.4844 \\ 
	& \multirow{2}{*}{2,10} & Unfair for A & \multirow{2}{*}{6E5k} & 0.5005 \\
	& & Unfair for B & & 0.4995 \\
	\midrule
	\multirow{4}{*}{\specialcell{$nand$\\(N,K)}} & 20,2 & \multirow{4}{*}{Reliable} & 155k & 0.4129 \\ %
	& 20,3 & & 232k & 0.4685 \\
	& 60,1 & & 4717k & 0.2695 \\
	& 60,3 & & 1.4E5k & 0.6377 \\
	\midrule
	\multirow{3}{*}{\specialcell{$SWaT$\\(S,R)}} & \multirow{1}{*}{5,1} & \multirow{3}{*}{Underflow} & \multirow{1}{*}{$\infty$} & \multirow{1}{*}{0.1713} \\

	 & \multirow{1}{*}{5,5} &  & \multirow{1}{*}{$\infty$} & \multirow{1}{*}{0.1389} \\

	 & \multirow{1}{*}{10,5} &  & \multirow{1}{*}{$\infty$} & \multirow{1}{*}{0.3333} \\
	\bottomrule	
	\end{tabular}
	\label{tb:cases}
\end{table}

\begin{table}[t]
\centering
\caption{Parameters for each algorithm.}
\begin{tabular}{@{}ccc@{}}
\toprule
Algorithm & Parameter           & Value \\ \midrule
AA        & $\epsilon_{max}$             & 64    \\
\midrule
\multirow{4}{*}{GA}       & number of chromosome &      50  \\
          & maximum generation  &       10 \\ 
          & mutation rate       &   0.1    \\
          & selection probability&     0.9  \\
          
\midrule
\multirow{4}{*}{LAR}       & $\alpha$               & 0.05  \\
          & $\beta$                & 0.05  \\
          & $\sigma$               & 0.05  \\
          & $\epsilon_{max}$             & 64     \\ 
          & $\varsigma$ & 0.8 \\
          \bottomrule
\end{tabular}
\label{tb:para}
\end{table}

Table~\ref{tb:para} shows the details of the parameters we used in our experiments for AA, GA and LAR respectively. For AA, there is only one parameter which is the maximum value of $\epsilon$ to choose from\footnote{Note that the minimum value is 1 as required for convergence}. For GA, there are 4 parameters as follows. We set the number of chromosomes in each generation to be 50. The mutation rate for the mutation operator is set to be 0.1. We use the tournament selection strategy to select the good chromosomes in each generation and the probability to keep the winner is set to be 0.9. Finally, the maximum number of generation for the evolution is set to be 10. For LAR, the parameters are mainly for hypothesis testing, the maximum value of $\epsilon$ for AA, and the minimum classification accuracy for SVM. In particular, the error bounds for hypothesis testing are set as $\alpha=0.05,\beta=0.05$. The indifference region is set to be $\sigma=0.05$. The maximum value of $\epsilon$ to choose from is set to be 64 which is the same as AA. The minimum required accuracy for SVM classification $\varsigma$ is set to be 0.8. All the experiments were conducted on an OS X machine with 2.6GHz quad-core Intel Core i7 processor and 8 GB RAM.
We aim to answer the following research questions through our experiments.

\subsubsection*{RQ1: Is LAR better in verifying systems than state-of-the-art approaches based on learning probabilistic models?}

The results are summarized in Table~\ref{tb:exp}, where - means no results after the preset timeout 2 hours. We show the results of applying AA and GA on the basic abstraction which is defined by whether a state satisfies the target property or not in column `AA Basic' and `GA Basic' respectively. We can observe that in all cases the verification results on the learned models are 1 due to overly coarse abstraction and inappropriate state merging. 
Column `AA' shows the results of the AALERGIA algorithm and column `GA' shows the results of the GA-based approach. Column `LAR' shows our results. We leave the explanation of column `AA-LAR' later. Note that for $crowds$ and $nand$ protocol, we have to do manual abstraction first for AA and GA, i.e., heuristically select a smallest number of variables which gives us reasonable results, otherwise neither AA nor GA can learn a model within the time limit. This means that LAR is competing with AA and GA combining some manual abstraction. We remark that this is precisely the advantage of LAR which automatically identifies a level of abstraction which verifies the property.   

\begin{table*}[t]
\centering
\caption{Verification results on learned models using different learning approaches.}
\label{tb:exp}
\begin{tabular}{@{}ccc|cc|cc|cc|cc|ccc|cc@{}}
\toprule
\multirow{2}{*}{System} & \multirow{2}{*}{\#Config} & \multirow{2}{*}{$P_{act}$}  & \multicolumn{2}{c}{AA Basic}   &  \multicolumn{2}{c}{AA}     &  \multicolumn{2}{c}{GA Basic} &  \multicolumn{2}{c}{GA}    &  \multicolumn{3}{c}{LAR}    &  \multicolumn{2}{c}{AA-LAR} \\ \cline{4-16}
&&& \#states & result & \#states & result & \#states & result & \#states & result & iters & \#states & result & \#states & result \\
\midrule
\multirow{4}{*}{$crowds$} & 1      & 0.1458 &    4      & 1 & 10572 & 1 & 3 &     1     & 102 & 1 & 2 & 6 & 0.1549 & 5 & 0.156  \\
 & 2      & 0.1048 &     4     & 1 & 10894 & 1 & 3 &     1     & 158 & 0.1054 & 2 & 6 & 0.101 & 5 & 0.114  \\
 & 3      & 0.0922 &     4    & 1 & - & -      & 3 &     1     & 206 & 1 & 2 & 6 & 0.066 & 5 & 0.082  \\
 & 4      & 0.0861 &     4     & 1 & - & -      & 3 &     1     & 266 & 1 & 2 & 6 & 0.0877 & 5 & 0.084  \\
 \midrule
\multirow{4}{*}{$egl$}    & \multirow{2}{*}{1}      & 0.5156 &    \multirow{2}{*}{4}      &  \multirow{2}{*}{1} & \multirow{2}{*}{-} & \multirow{2}{*}{-}      & \multirow{2}{*}{3} &  \multirow{2}{*}{1}        & \multirow{2}{*}{-} & \multirow{2}{*}{-}      & \multirow{2}{*}{1} & \multirow{2}{*}{34} & 0.4961 &  \multirow{2}{*}{55} & 0.4834 \\
    &       & 0.4844 &          &  &  &       &  &          &  &       &  &  & 0.5039 & & 0.5166 \\
    & \multirow{2}{*}{2}      & 0.5005 &    \multirow{2}{*}{4}      & \multirow{2}{*}{1} & \multirow{2}{*}{-} & \multirow{2}{*}{-}      & \multirow{2}{*}{3} & \multirow{2}{*}{1}         & \multirow{2}{*}{-} & \multirow{2}{*}{-}      & \multirow{2}{*}{1} & \multirow{2}{*}{4} & 0.4619 & \multirow{2}{*}{107} & 0.4888 \\
    &       & 0.4995 &          &  &  &       &  &          &  &       &  &  & 0.5381 &  & 0.5112 \\
    \midrule
\multirow{4}{*}{$nand$}   & 1      & 0.4129 &     4     & 1 & 1257 & 0.4878 & 3 &     1     & 90 & 0.4999   & 3 & 8 & 0.3844 & 7 & 0.3849 \\
   & 2      & 0.4685 &     4     & 1 & 1793 & 0.4995 & 3 &    1      & 90 & 0.4995 & 3 & 8 & 0.5332 & 7 & 0.43 \\
   & 3      & 0.2695 &     4     & 1 & 14987 & 0.1428 & 3 &   1       & 250 & 0.1707 & 2 & 6 & 0.1458 & 5 & 0.2286 \\
   & 4      & 0.6377 &     4     & 1 & - & -      & 3 &     1     & - & -      & 4 & 11 & 0.606   & 9 & 0.6      \\
   \midrule
\multirow{4}{*}{$SWaT$}   & 1      & 0.1713 &     101     & 1 & - & -      & 3 &    1      & - & -      & 7 & 274 & 0.2 & 852 & 0.2 \\
   & 2      & 0.1389 &    186      & 1 & - & -      & 2 &    1      & - & -      & 12 & 930 & 0.1408 & 4139 & 0.189   \\
   & 3      & 0.3333 &    255      & 1 & - & -      & 3 &    1      & - & -      & 20 & 974 & 0.2268 & 4226 & 0.2999 \\ \bottomrule
\end{tabular}
\end{table*}

\begin{table*}[t]
	\centering
	\caption{Experiment results of LAR: given different safety threshold in the property
	} 
	\begin{tabular}{ccc|cccc|cccc|cccc}
	\toprule
	\multirow{2}{*}{System} & \multirow{2}{*}{\#Config} & \multirow{2}{*}{$P_{act}$} & \multicolumn{4}{c}{LAR-$r_2$}  & \multicolumn{4}{c}{LAR-$r_3$} & \multicolumn{4}{c}{LAR-$r_4$} \\ 
	\cline{4-15}
	 & & & \#states & iters & time & result & \#states & iters & time & result & \#states & iters & time & result\\
	\midrule
	\multirow{4}{*}{$crowds$} 
		& 1 & 0.1458 & 6 & 2 & 109 & 0.1723 & 6 & 2 & 104 & 0.1599 & 6 & 2 & 96  & 0.1794 \\
	       & 2 & 0.1048 & 6 & 2 & 119 & 0.1044 & 6 & 2 & 119 & 0.1051 & 6 & 2 & 108 & 0.1051 \\
	       & 3 & 0.0922 & 6 & 2 & 114 & 0.063  & 6 & 2 & 110 & 0.0594 & 6 & 2 & 104 & 0.0561 \\
	       & 4 & 0.0861 & 6 & 2 & 116 & 0.0845 & 6 & 2 & 115 & 0.0813 & 6 & 2 & 121 & 0.0848 \\
	\midrule
	\multirow{4}{*}{$egl$} & \multirow{2}{*}{1} & 0.5156 & \multirow{2}{*}{34} & \multirow{2}{*}{1} & \multirow{2}{*}{5} &  0.4961 & \multirow{2}{*}{34} & \multirow{2}{*}{1} & \multirow{2}{*}{5} & 0.4961 & \multirow{2}{*}{34} & \multirow{2}{*}{1} & \multirow{2}{*}{5} & 0.4961\\
	& & 0.4844 & & &  & 0.5039 & & & & 0.5039 & & & & 0.5039\\ 
	& \multirow{2}{*}{2} & 0.5005 & \multirow{2}{*}{4} & \multirow{2}{*}{1} & \multirow{2}{*}{6} &  0.4619 & \multirow{2}{*}{4} & \multirow{2}{*}{1} & \multirow{2}{*}{6} & 0.4619 & \multirow{2}{*}{4} & \multirow{2}{*}{1} & \multirow{2}{*}{6} & 0.4619\\
	& & 0.4844 & & &  & 0.5381 & & & & 0.5381 & & & & 0.5381\\
	\midrule
	\multirow{4}{*}{$nand$} & 1 & 0.4129 & 8  & 3 & 119 & 0.4532 & 8  & 3  & 95 & 0.3674 & 8  & 3  & 83 & 0.349  \\
	     & 2 & 0.4685 & 8  & 3 & 0.6 & 0.4912 & 8  & 3  & 88 & 0.5098 & 8  & 3  & 60 & 0.4997 \\
	     & 3 & 0.2695 & 6  & 2 & 100 & 0.1276 & 6  & 2  & 80 & 0.1333 & 6  & 2  & 76 & 0.1818 \\
	     & 4 & 0.6377 & 11 & 4 & 166 & 0.381  & -- & -- & -- & --     & -- & -- & -- & --     \\
	\midrule
	\multirow{3}{*}{$SWaT$} & 1 & 0.1713 & 274 & 7  & 48  & 0.2    & 204 & 6  & 40  & 0.375  & 204 & 6  & 41  & 0.375  \\
	     & 2 & 0.1389 & 693 & 11 & 121 & 0.2289 & 542 & 9  & 84  & 0.3932 & 532 & 7  & 56  & 0.4226 \\
	     & 3 & 0.3333 & 910 & 17 & 719 & 0.4475 & 910 & 17 & 738 & 0.4475 & 910 & 17 & 723 & 0.4475 \\
	\bottomrule	
	\end{tabular}
	\label{tb:exp2}
\end{table*}

It can be observed that LAR successfully verified all systems, whereas both AA and GA failed on about half the systems. In most cases especially for the benchmark systems, LAR only takes few iterations to identify a model which allows us to prove the property. This is particularly useful in verifying systems with a large number of (or infinite) states. For instance, for the $egl$ model, while AA (or GA) fails to learn or verify any of the 4 cases, LAR successfully verifies all of them in about 5 seconds and the learned models have 4 states only. The latter two $crowds$ cases show similar results. This is especially the case for the SWaT system. 
Because the variables in this system are all float numbers, without abstraction, every logged state is different from others and as a result, there is no generalization (i.e., no state merging in AA and GA) and no models could be learned.
On the other hand, LAR is able to learn reasonably small models by automatic abstraction refinement.

Efficiency-wise, Fig.~\ref{fig:learn-time} shows the time cost of AA, GA and LAR for different systems. We can observe that LAR takes much less time than AA and GA (given the same number of traces). In over half of the cases neither AA nor GA can learn a model within the time limit, whereas LAR can learn an accurate model within minutes. This is largely because a much smaller DTMC is learned, since the abstract traces only have a few symbols (up to $2^{\#P}$ where $\#P$ is the number of predicates in $P$) in the alphabet. Furthermore, because the learned model is small, model checking whether the model satisfies the property through PMC takes much less time. Fig.~\ref{fig:lar-time} shows the time distribution of each step in LAR, which include learning the model (Learning), PMC on the learned model (PMC), hypothesis testing of counterexample (Hypothesis testing) and abstraction refinement (Refinement) to generate a new predicate. A close look reveals that learning, PMC and spuriousness checking through hypothesis testing dominates LAR's time (about 80\% on average), while abstraction refinement takes very little time, i.e., mostly less than 10\%. Fig.~\ref{fig:iter-time} shows the time cost of LAR in each iteration. It can be observed that in most cases the time cost grows slowly by iteration. The only exception is the last iteration of swat-2. We observe that there is a significant increase in the number of states in the learned model in the last iteration, which takes a relatively long time to model check. We thus have the following answer to RQ1.
\begin{framed}
\noindent \emph{Answer to RQ1: LAR is more effective and efficient in verifying complex systems than state-of-the-art approaches based on learning probabilistic models by automatically identifying a right level of abstraction for learning.
}\end{framed}


\subsubsection*{RQ2: Is LAR able to learn models better than those learned by state-of-the-art approaches?} 

Recall that if LAR reports that the property is verified, a model is returned. It is thus important that the learned models are of high-quality. In the following, we evaluate the quality of the learned models in three aspects. 

First of all, it can be observed from Table~\ref{tb:exp} that {\it the models generated by LAR have much fewer states than AA or GA, i.e., often order of magnitudes fewer than the model learned by AA.} The reason is that we always start with learning based on the coarsest abstraction and only add one predicate for abstraction if necessary. We remark that learning small models is important if the models are to be reviewed by experts or used to implement runtime monitors as they are easier to interpret. For instance, for the example $crowds$ protocol, the learned new predicate is $new<runCount$, where $runCount$ is a variable representing how many runs are left and $new$ represents whether it is about to start a new run. Intuitively, this predicate allows us to separate the last run from the other runs, which is relevant because if $observer0$ is 0 after the second last run, it is impossible to reach a state satisfying $observe0 > 1$. With the learned predicate, we can easily interpret the learned system models and implement a runtime monitor monitoring how the system evolves afterwards. Secondly, to evaluate the accuracy of the learned models, we apply PMC on the model learned by LAR to compute the probability $\mathcal{P}(\textbf{F} \varphi)$ (although LAR's verification results are either true or false). It can be observed that {\it the models produced by LAR often (i.e., in 14 out of 15 cases) have more or equally precise verification results than those models learned by both AA and GA.} 
Lastly, we measure how accurate LAR's models are, compared to `best' model which can be learned at the learned abstract level. To do that we generate a large number of traces (5 times more than those used in LAR), abstract them with the same set of predicates learned by LAR, and apply AA to learn a model. Given that AA is guaranteed to converge~\cite{AA,AAJ}, we assume these models to be the accurate models at this abstract level. We then compare them with LAR's models.
Column `AA-LAR' shows details related to these accurate models learned by AA, i.e., the number of states and the result of $\mathcal{P}(\textbf{F} \varphi)$ based on the models. It can be observed that the models learned by LAR are rather close to these best models. It thus suggests that {\it not only LAR has found an abstract level to prove the property but also it is able to learn accurate models at the chosen abstract level.} We thus have the following answer to RQ2.
\begin{framed}
\noindent \emph{Answer to RQ2: LAR is able to learn more accurate and much smaller models (at the identified abstraction level) than state-of-the-art approaches.}
\end{framed}

\begin{figure}[t]
\centering
\includegraphics[width=.5\textwidth]{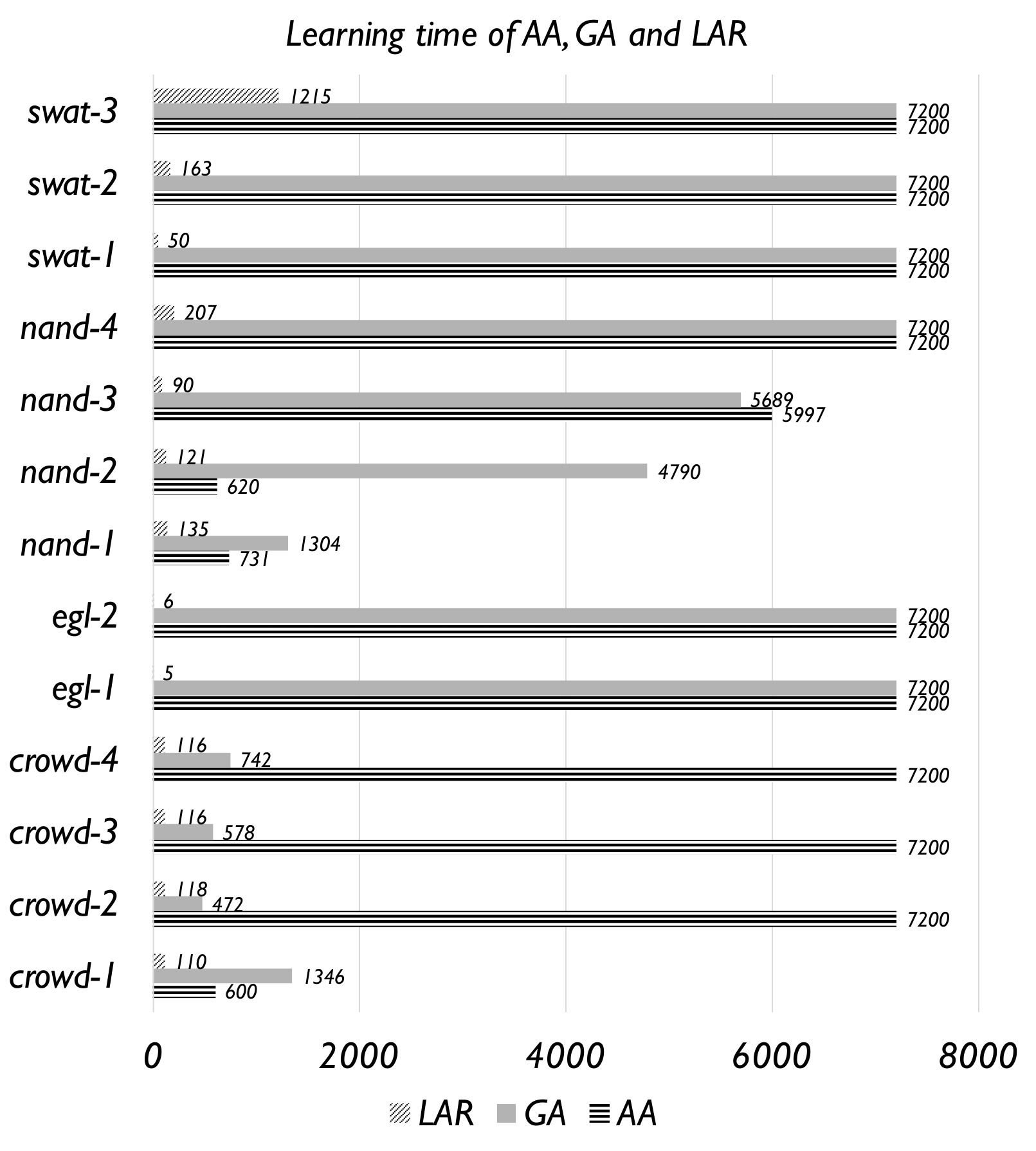}
\caption{Learning time of AA, GA and LAR in seconds. The preset timeout is 7200 s.}
\label{fig:learn-time}
\end{figure}

\begin{figure}[t]
\centering
\includegraphics[width=.5\textwidth]{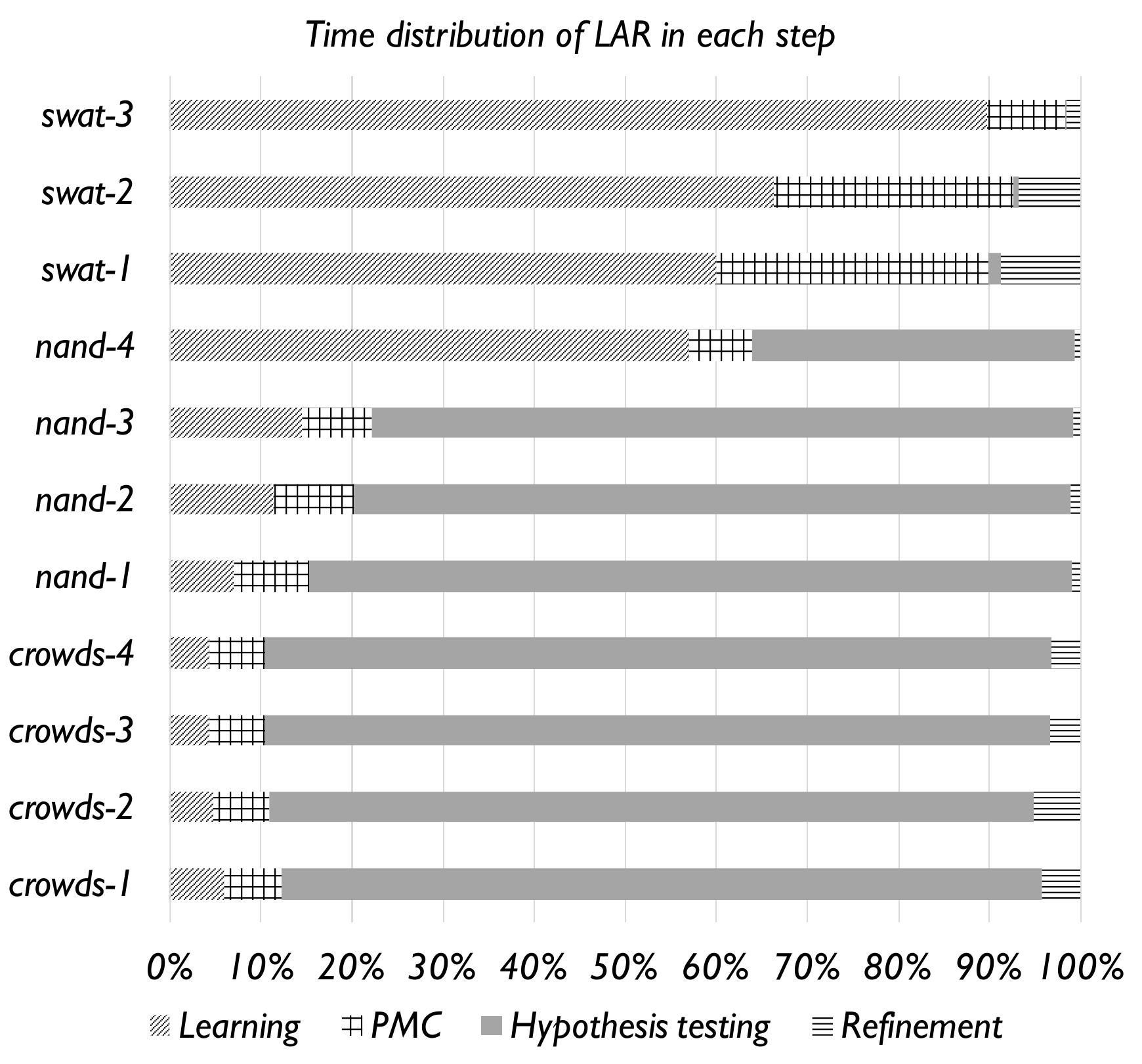}
\caption{The time distribution of LAR in each step including learning, PMC, hypothesis testing and refinement.}
\label{fig:lar-time}
\end{figure}

\begin{figure}[t]
\centering
\includegraphics[width=.5\textwidth]{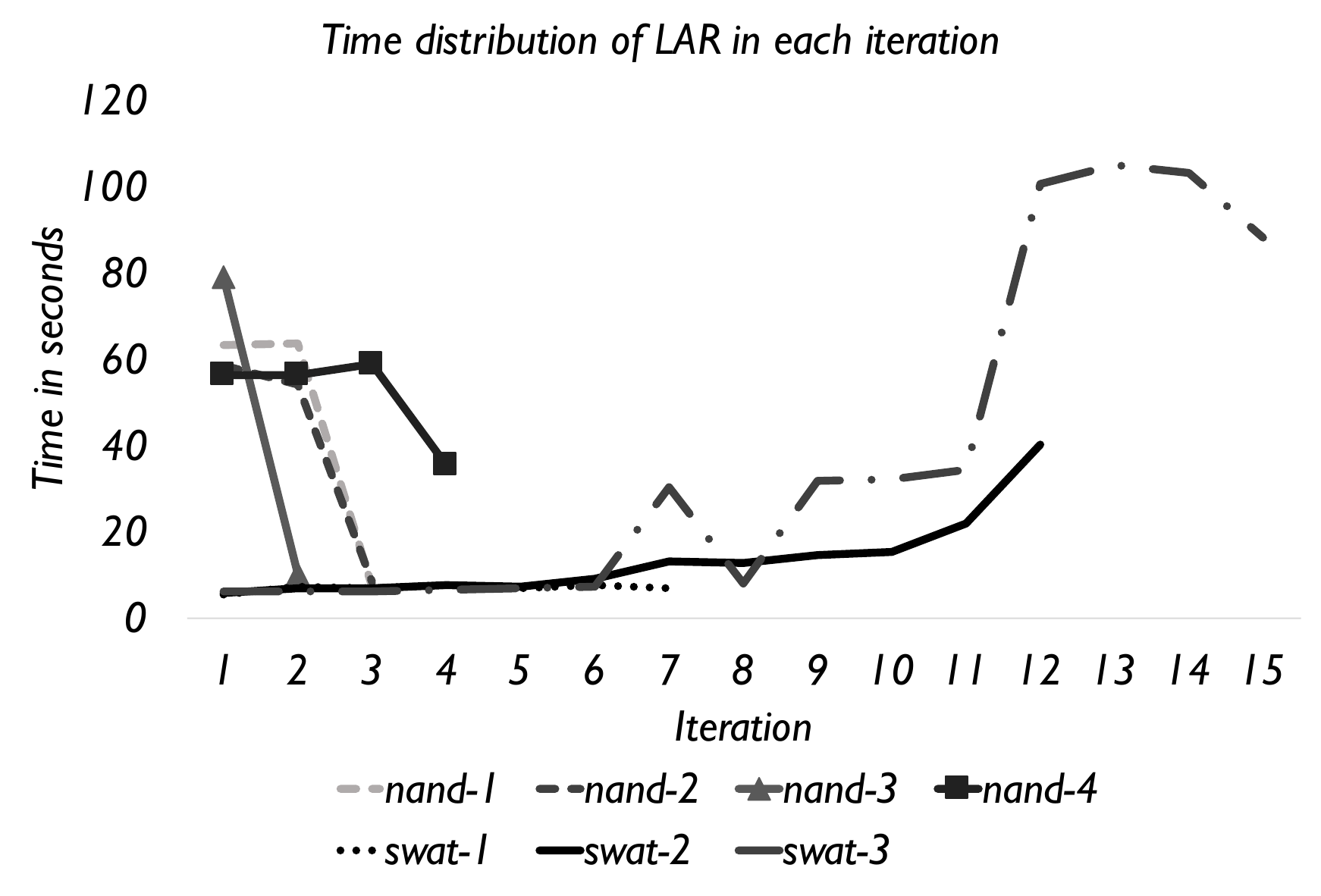}
\caption{The time distribution of LAR in each iteration.}
\label{fig:iter-time}
\end{figure}

\subsubsection*{RQ3: Is LAR able to learn models in a property-guided way?} 

Recall that LAR aims to learn a model which is sufficiently accurate to show that the probability of reaching certain states is less than $r$. Thus, with a larger $r$, ideally a coarser model which is easier to comprehend will be learned. In the above experiments, we fix $r$ to be 20\% more than the actual probability. In the following, we perform additional experiments to show the effect of having different $r$ values. That is, we run $r$ which are 40\%, 80\% and 100\% above the actual probability (for SWaT configuration 1 and 2, we run $r$ which are 2/3/4 times of the actual probability since $P_{act}$ is a small value). The results are shown in Table~\ref{tb:exp2} (-- means no results because the threshold is larger than 1.). 

It can be observed that for benchmark systems, similarly or identical models are learned by LAR for different $r$-values. The reason is that for these relatively small systems, it is easy to learn a model which produces accurate verification results after few iterations. In contrast, for the SWaT system, LAR learns a coarser model with fewer states, and less accurate result (which is still good enough to prove the property). We thus have the following answer to RQ3.
\begin{framed}
\noindent \emph{Answer to RQ3: LAR is able to learn models at different abstraction levels in a property-guided way.}
\end{framed}

\mypara{Threats to validity} The above experiments successfully show the effectiveness and efficiency of LAR over the test subjects. In the following, we discuss several threats to the validity. \textit{First of all,} it can be observed that a strict property (with probability close to the actual probability) requires learning a detailed model, whereas a loose property requires a coarse model. In the worse case, if the property can only be verified with all details in the system, then we must refine all the way until every detail is included. One question is then how we know whether the property is strict or not so as to tell whether LAR will be effective. In general, there is no good solution. \textit{Secondly,} there is a tradeoff between the complexity of the classification algorithm to generate the predicate and the interpretability of the learned model. For example, if the dataset is not linear separable and we apply kernel functions for SVM, the generated predicate will not be easy to interpret. For real-world complex systems, it might be possible that a linear classifier does not exist, in which case we will either fail to learn a model (without adopting kernels for SVM) or learn a model that is not intuitive for understanding (adopting kernels for SVM). \textit{Thirdly,} LAR refines the abstraction by choosing a particular spurious transition and separate the concrete states there. Thus, the selection of different spurious transitions may lead to different ways of refinement and thus models of different quality. \textit{Lastly,} LAR uses statistical hypothesis testing to check the spuriousness of a counterexample in every iteration. An alternative way is to apply SMC to check the given property directly, learn a model from the traces captured during SMC if the property is violated and explain to the user how the system fails. It is in general hard to predict which approach is more efficient because it depends on the specific property and system. For instance, in case the property is indeed violated by the system, the alternative approach might be more efficient in triggering such violations. However, the improvement depends on the difference between the probability threshold in the property and the actual probability of property violation. The larger the difference, the more significant is the improvement. In reality, CPS are usually equipped with safety mechanisms to avoid safety violations, which means that the probability of the actual safety violation is usually low. In case the property is satisfied, the comparison boils down to comparing the efficiency of spuriousness checking and the efficiency of SMC, which is also hard to tell. According to our experiments, a spurious counterexample can usually be checked quickly using SPRT during each iteration due to the wrong abstraction. On the other hand, SMC for a property with low probability may take a long time. In the future, we will evaluate them in a separate work. 


\section{Related Work} \label{related}
This work takes a further step towards the emerging trend to leverage system data for formal verification. It extends the recent line of work to `learn' (probabilistic especially) models (e.g., in the form of DTMC, CTMC, stationary models and MDPs) from system data for model checking~\cite{AA,AAJ,sen2004learning,chen2012learning,mao2012learning} in order to avoid manual modeling. Existing learning algorithms are often based on algorithms designed for learning (probabilistic) automata, as evidenced in~\cite{ron1996power,ron1995learnability,carrasco1994learning,de2010grammatical,carrasco1999learning,angluin1987learning}. In~\cite{brazdil2014verification}, reinforcement learning algorithms are adopted to verify Markov decision processes, without constructing explicit models.
LAR complements these model learning approaches with a CEGAR-style framework so that model learning can be done in a more guided way, i.e., learn at the proper level of abstraction towards the verification of certain properties. 
On the other hand, rather than learning a single model from system data, the authors of~\cite{haesaert2017data,abate2017formal} proposes a Bayesian inference approach which performs verification over the entire property- and model-based feasible class of models while providing a confidence over the data-generating system.

This work is also closely related to the following line of work which aims to quantitatively verify stochastic hybrid systems with confidence. In~\cite{haesaert2017verification,abate2015adaptive,soudjani2015quantitative,esmaeil2013adaptive}, the authors uses formal abstraction to quantitatively approximate the system with an abstract computable model. The approximation is parameterized, which then allows to quantify the distance between the models as well as their solution processes~\cite{zamani2014symbolic}. Furthermore, a CEGAR-style procedure is proposed to refine the abstraction when the approximate model is not satisfactory~\cite{tkachev2017quantitative}.     

SMC~\cite{younes2002probabilistic,sen2004statistical,younes2011statistical,rohr2013simulative,clarke2011statistical} is another line of work for quantitative verification that can be applied when the system model is not available. Hypothesis testing is initially adopted by SMC mainly for bounded properties.
 There is some recent work on extending SMC to unbounded properties~\cite{younes2011statistical,rohr2013simulative}, and
  combining learning and abstraction for faster SMC~\cite{learningAbstractionSMC}. The main difference between LAR and SMC is that LAR generates models as a part of the verification results, which would offer knowledge or insight on how the system works and why the property is verified. Furthermore, because LAR verifies the system based on the learned model, it is not limited to bounded properties.

The main idea of this work is inspired by CEGAR~\cite{cegar,clarke2003counterexample,henzinger2002lazy,henzinger2004abstractions,kroening2004counterexample} and in particular its extension to probabilistic systems. The fundamental questions and pragmatic issues of probabilistic abstraction refinement are first explored in~\cite{probcegar} in the context of predicate abstraction~\cite{graf1997construction,wachter2007probabilistic}, where a probabilistic counterexample (obtained using the approach documented in~\cite{cegeneration}) is analyzed to refine an abstraction. Compared to CEGAR and probabilistic CEGAR, LAR is different as LAR does not require any user-provided system models, which yields a completely automatic verification framework from system data.

\section{Conclusion}\label{con}
In this work, we propose a framework to automatically verify discrete-time complex systems without manual modeling through a combination of learning, abstraction and refinement. Our evaluation shows that LAR not only is more effective and efficient in verifying systems than state-of-the-art learning approaches by automatically identifying a level of abstraction in a property-guided way, but also generates models of high quality which could be useful for subsequent system analysis purposes like runtime monitoring or model-based testing.
Our main contribution lies in proposing such an automatic verification framework from system data, a systematic way to analyze spurious counterexamples for refinement, and adopt SVM to generate new predicates from system data directly without the assumption of having the original system model. In the future, we plan to compare our approach with statistical model checking by conducting a more systematic empirical study and extend the work to Markov Decision Process to support modeling and verification of a richer range of systems.

\bibliographystyle{plain}
\bibliography{bib/arf}
\appendices
\section{Bound on the learning error}\label{app1}
Let $A=(a_{ij})_{1\leq i,j\leq m}$ be a stochastic matrix of size $m\times m$, which represents the actual transition matrix of $M$ and $\hat{A}=(\hat{a}_{ij})_{1\leq i,j\leq m}$ be that of $D_P$. For simplicity, we assume $\hat{A}$ is learned through a simple algorithm based on Monte Carlo frequency estimate. That is, $\hat{a}_{ij}=n_{ij}/n_i$ where $n_i$ is the number of times a transition has been taken from state $i$ and $n_{ij}$ the number of times that transition from state $i$ to state $j$ has been taken. We denote $n_0=\min_{1\leq i\leq m} n_i$.

Given a precision bound $\epsilon$ and a confidence bound $\delta$, we show that it is feasible to determine a sampling scheme that guarantees after a minimal number of samples:
\begin{equation}
\label{obj}
P(||\hat{A}-A||>\epsilon)<\delta
\end{equation}
where $||\hat{A}-A||$ denotes a matrix distance between $\hat{A}$ and $A$.

In what follows, let the distance be defined as $||\hat{A}-A||=\sup_{ij}|\hat{a}_{ij}-a_{ij}|$.

\begin{eqnarray*}
P(||\hat{A}-A||>\epsilon)&=&P(\sup_{ij}|\hat{a}_{ij}-a_{ij}|>\epsilon)\\
&=& P\left(\bigcup_{ij}\,\lbrace|\hat{a}_{ij}-a_{ij}|>\epsilon\rbrace\,\right)\\
&\leq& \sum_{ij} P(|\hat{a}_{ij}-a_{ij}|>\epsilon)\\
&\leq& \sum_{ij} P(|\frac{n_{ij}}{n_i}-a_{ij}|>\epsilon)\\
&\leq& \sum_{ij} 2\exp(-2n_i\epsilon^2)\quad\text{(Okamoto bound)} \\
&\leq& 2m^2\exp(-2n_0\epsilon^2)
\end{eqnarray*}

Based on the above, given $\delta$ and $\epsilon$, $\Pi$ must contain enough traces such that each state is visited at least $n_0=\lceil \frac{1}{2\epsilon^2}\log\left(\frac{\delta}{2m^2}\right)\rceil$ if we are to guarantee that the learned model satisfies the constraint stated in Equation~\ref{obj}.

This bound is conservative. However, this sampling scheme terminates with probability $1$ since all the states are reachable. Furthermore, there are several ways to improve this bound. First, the term $m^2$ could be replaced by the number of transitions $K$ such that $a_{ij}$ is different from $0$ or $1$. In that case, $P(|\hat{a}_{ij}-a_{ij}|>\epsilon)$ is necessarily equal to $0$, then sequential frequentist methods based on sharper bounds may be used (e.g. \cite{QEST17}) instead of the used Okamoto bound. Besides, the number of samples is also impacted by the choice of $\epsilon$. 
 Finally, the choice of the matrix distance is also important. For example, at the learning level, we could prefer a distance that takes into account probabilities $p_i$ of being in state $i$ in order to guarantee a fine per-transition error bound over states that are often visited and a coarser bound over states that are not often visited. This is left to future work.
%

\begin{IEEEbiography}[{\includegraphics[width=1in,height=1.25in,clip,keepaspectratio]{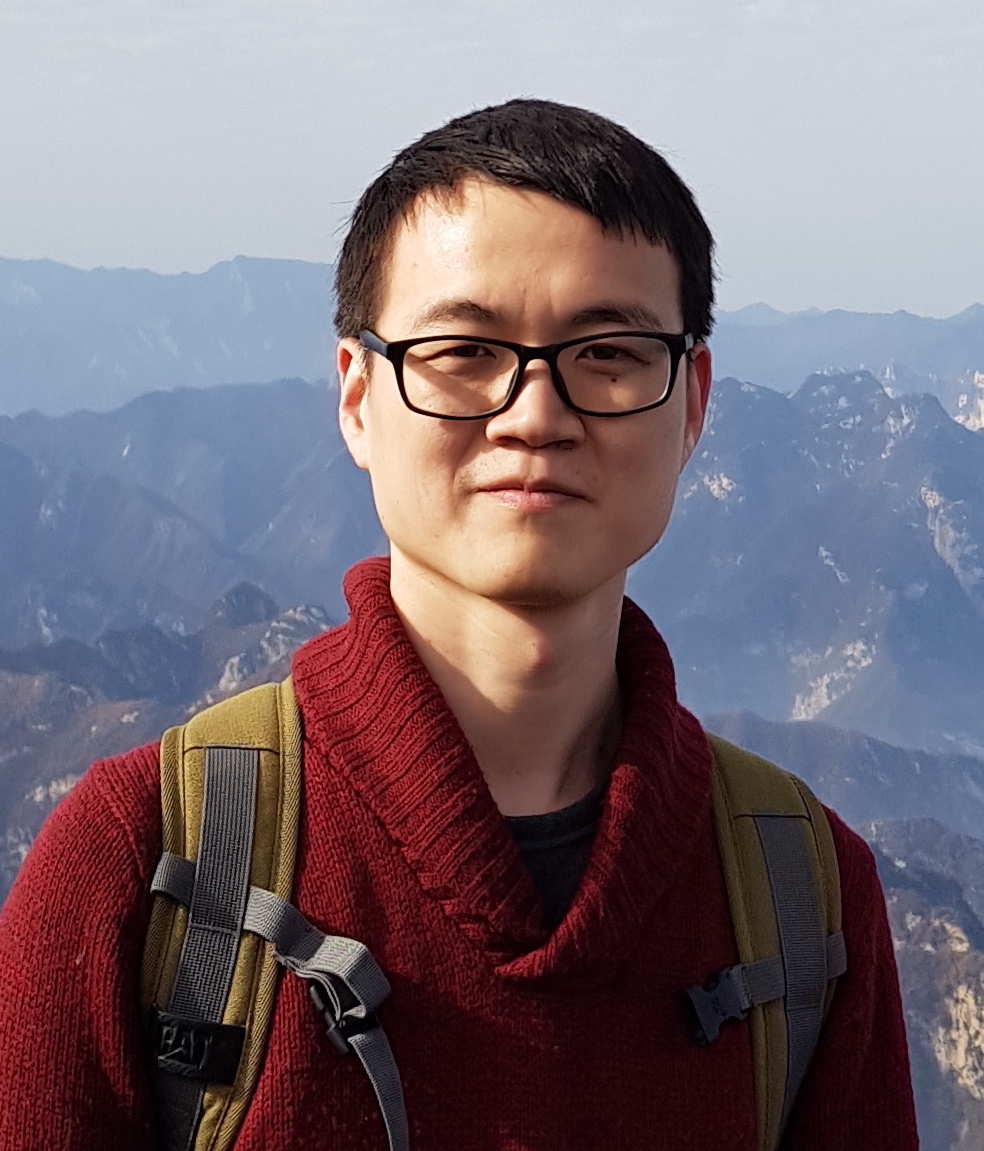}}]{Jingyi Wang}
is currently a research fellow in the pillar of Information Systems Technology and Design at Singapore University of Technology and Design, Singapore, where he obtained his PhD in March 2018. Before this, he got his Bachelor of Engineering degree in Information Engineering from Xi'an Jiaotong University, Xi'an, China in 2013. His main research interests are to apply machine learning and statistical techniques into formal verification as well as security problems in deep learning systems.
\end{IEEEbiography}

\begin{IEEEbiography}[{\includegraphics[width=1in,height=1.25in,clip,keepaspectratio]{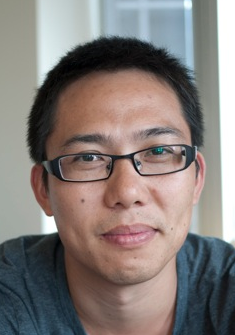}}]{Jun Sun} is currently an associate professor at Singapore University of Technology and Design (SUTD). He received Bachelor and PhD degrees in computing science from National University of Singapore (NUS) in 2002 and 2006. In 2007, he received the prestigious LEE KUAN YEW postdoctoral fellowship. He has been a faculty member of SUTD since 2010. He was a visiting scholar at MIT from 2011-2012. Jun's research interests include software engineering, formal methods, program analysis and cyber-security. He is the co-founder of the PAT model checker. 
\end{IEEEbiography}


\begin{IEEEbiography}[{\includegraphics[width=1in,height=1.25in,clip,keepaspectratio]{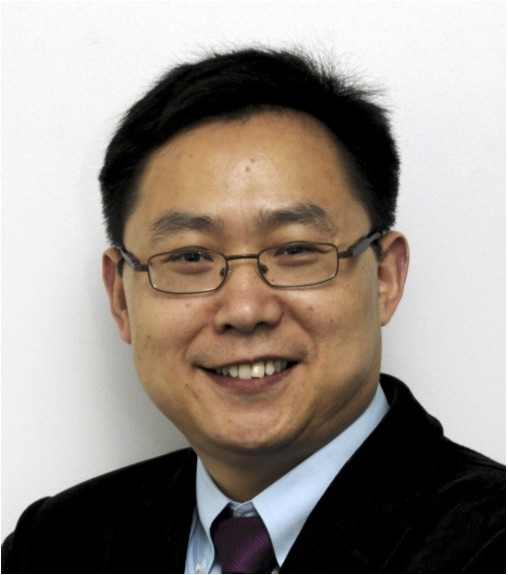}}]{Shengchao Qin} got his PhD in Applied Mathematics from Peking University and also worked as a Postdoctoral
Research Fellow in National University of Singapore under the Singapore-MIT Alliance program, before moving
his job to UK. His research interests lie mainly in formal methods, software engineering and programming languages,
in particular, formal specification and modelling, program analysis and verification, theories of programming, program
logic such as separation logic. To this date he has published over 90 papers in international journals and peer-refereed
international conferences.
\end{IEEEbiography}

\begin{IEEEbiography}[{\includegraphics[width=1in,height=1.25in,clip,keepaspectratio]{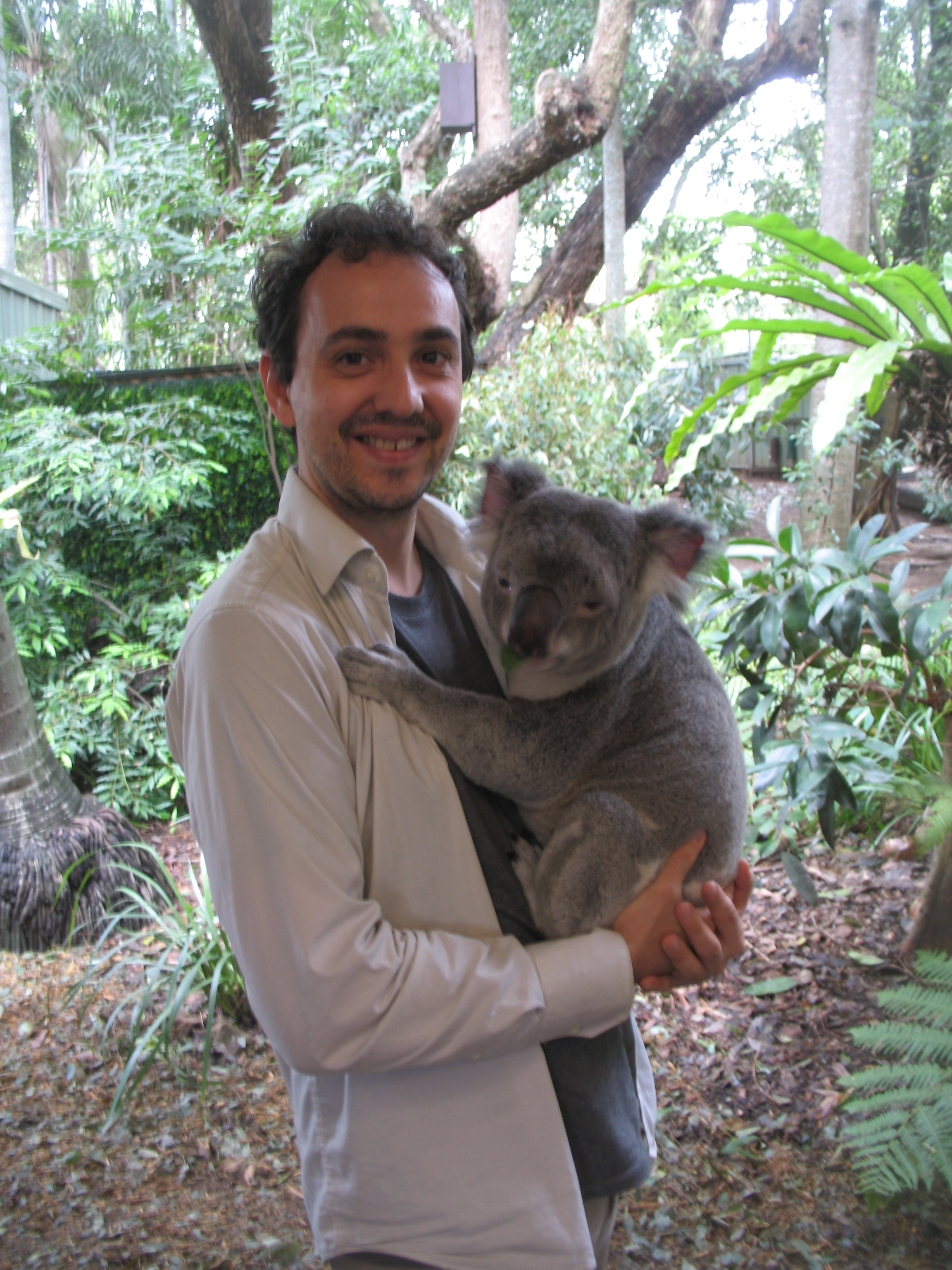}}]{Cyrille Jegourel} was born in France, on October 22, 1984. He received the B.S. degree in mathematics in 2005, the M.S. degree in statistical engineering in 2008, and the Ph.D. degree in computer science in 2014 from the University of Rennes 1, France. He was a development and research engineer at Inria Rennes, Bretagne Atlantique, France, from time to time between 2008 and 2015. He was a postdoctoral Research Fellow at the National University of Singapore, Singapore, in 2015. In 2016, he joined the Singapore University of Technology and Design, Singapore, where he was engaged as a Research Fellow to work on sampling techniques and statistical aspects for formal methods.
\end{IEEEbiography}
\vfill

\end{document}